\documentclass[11pt,letterpaper]{article}
\usepackage[a4paper,left=3cm,right=3cm,top=3cm,bottom=3cm]{geometry}

\usepackage{times}
\usepackage{color}
\usepackage{helvet}
\usepackage{courier}
\usepackage{tabularx}
\usepackage{amsmath}
\usepackage{amsthm}
\usepackage{amssymb,amsfonts,textcomp}
\usepackage{latexsym}
\usepackage{graphicx} 
\usepackage{multirow}
\usepackage{rotating}
\usepackage{url}
\usepackage{color}
\usepackage{comment}
\usepackage{paralist}
\usepackage{booktabs}
\usepackage{bm}
\usepackage[]{units}
\usepackage{natbib}
\usepackage[noend, ruled, noline]{algorithm2e}
\usepackage{tikz}

\newtheorem{definition}{Definition}
\newtheorem{proposition}{Proposition}
\newtheorem{example}{Example}
\newtheorem{corollary}{Corollary}
\newtheorem{theorem}{Theorem}
\newtheorem{lemma}{Lemma}
\allowdisplaybreaks

\newcommand{\figref}[1]{Figure~\ref{#1}}

\newcommand{\thmref}[1]{Theorem~\ref{#1}}

\newcommand{\defref}[1]{Definition~\ref{#1}}

\newcommand{\exref}[1]{Example~\ref{#1}}

\newcommand{\ie}{i.e.,\xspace}

\RequirePackage{ifthen,calc}
\newsavebox{\ffbox}\newlength{\ffboxlen}
\newcommand{\todo}[1]{%
  {\sbox{\ffbox}{\textbf{TODO:}\ \textit{{#1}}\ \textbf{:ODOT}}
    \settowidth{\ffboxlen}{\usebox{\ffbox}}
		\addtolength{\ffboxlen}{-5mm}
    \ifthenelse{\ffboxlen>\linewidth}{%
      \noindent\marginpar{$>>>>$}\textbf{TODO:}\ \textit{{#1}}\ \textbf{:ODOT}\marginpar{$<<<<$}}{%
      \noindent\marginpar{$>><<$}\textbf{TODO:}\ \textit{{#1}}\ \textbf{:ODOT}}}}

\newcommand{\naturals}{{{\mathbb{N}}}}

\newcommand{\np}{{\mathrm{NP}}}

\newcommand{\vecw}{{\mathbf{w}}}

\makeatletter
\newtheorem*{rep@theorem}{\rep@title}
\newcommand{\newreptheorem}[2]{%
\newenvironment{rep#1}[1]{%
 \def\rep@title{#2 \ref{##1}}%
 \begin{rep@theorem}}%
 {\end{rep@theorem}}}
\makeatother
\newreptheorem{theorem}{Theorem}
\newreptheorem{proposition}{Proposition}
\newreptheorem{lemma}{Lemma}
\newreptheorem{corollary}{Corollary}

\title{Proportional Rankings}
\author{Piotr Skowron \and Martin Lackner \and Markus Brill \and Dominik Peters \and Edith Elkind}
\author{Piotr Skowron, Martin Lackner, Markus Brill, Dominik Peters, and Edith Elkind\\
Department of Computer Science\\
University of Oxford}

\date{}

\begin{document}

\maketitle

\begin{abstract}
In this paper we extend the principle of proportional representation to rankings.
We consider the setting where alternatives need to be ranked based on approval preferences.
In this setting, proportional representation requires that cohesive groups of voters are represented proportionally in each initial segment of the ranking.
Proportional rankings are desirable in situations where initial segments of different lengths may be relevant, e.g., hiring decisions (if it is unclear how many positions are to be filled), the presentation of competing proposals on a liquid democracy platform (if it is unclear how many proposals participants are taking into consideration), or recommender systems (if a ranking has to accommodate different user types).
We study the proportional representation provided by several ranking methods and prove theoretical guarantees. Furthermore, we experimentally evaluate these methods and present preliminary evidence as to which methods are most suitable for producing proportional rankings.
\end{abstract}

\section{Introduction}
\label{sec:introduction}

Consider a population with dichotomous (approval) preferences over a set of 300 alternatives. Assume that 50\% of the population approves of the first 100 alternatives, 30\% of the next 100 alternatives, and 20\% of the last 100 ones. Imagine that we want to obtain a ranking of the alternatives that in some sense reflects the preferences of the population. Perhaps the most straightforward approach is to use \emph{Approval Voting}, i.e., to rank the alternatives from the most frequently approved to the least frequently approved. A striking property of the resulting ranking is that half of the population does not approve any of the first 100 alternatives in the ranking.

In many scenarios, rankings obtained by applying Approval Voting are highly unsatisfactory; rather, it is desirable to interleave alternatives supported by different (sufficiently large) groups. For instance, consider an anonymous user using a search engine to find results for the query ``Armstrong''. Even if 50\% of the users performing this search would like to see results for Neil Armstrong, 30\% for Lance Armstrong, and 20\% for Louis Armstrong, it is not desirable to put only results referring to Neil Armstrong in the top part of the ranking shown on the first page; rather, results related to each of the Armstrongs should be displayed in appropriately high positions. % The above example illustrates that in some scenario it is necessary to have diversity in aggregated rankings.

There are numerous applications where such diversity within collective rankings is desirable. For instance, consider recommendation systems which aim at accommodating different types of users (the estimated preferences of these user types should be represented proportionally to their likelihood), a human-resource department providing hiring recommendations when it is unclear how many positions are to be filled, or committee elections with some additional structure (e.g., when we want to elect a committee with a chairman and vice-chairman, as well as substitute members). Another example application in which diversity in rankings is relevant is \emph{liquid democracy}~\citep{BKNS14a}. A defining feature of a liquid democracy is that all participants are allowed---and encouraged---to contribute to the decision making process. In particular, in a context where one (or more) out of several competing alternatives needs to be selected, each participant can propose their own alternative if they are not satisfied with the existing ones. This may lead to situations where a very large number of proposals needs to be considered. Since it cannot be expected that every participant studies all available alternatives before making a decision, the order in which competing alternatives are presented plays a crucial role \citep[][Chapter 4.10]{BKNS14a}.

Indeed, the idea of diversified rankings with respect to the preferences of a population appears in several literatures. For example, in the context of search engines, this aim is often referred to as \emph{diversifying search results}~\citep{conf/www/WelchCO11,journals/ftir/SantosMO15,conf/websci/KingraniLZ15,conf/waim/WangLY16}. 
% and seems to already be in use by some of major companies~\citep{conf/waim/WangLY16}; 
There are also models which incorporate this idea into online advertising (see the work of \citealp{Hu:2011:CSI:1963405.1963412}, and the references therein). 
In the context of liquid democracy, \citeauthor{BKNS14a} observe that using AV gives rise to what they call the ``noisy minorities'' problem: relatively small groups of very active participants can ``flood'' the system with their contributions, creating the impression that their opinion is much more popular than it actually is. This is problematic insofar as other alternatives (that are potentially much more popular) run the risk of being ``buried'' and not getting sufficient exposure. \citeauthor{BKNS14a} suggest that in order to prevent this problem, the ranking mechanism needs to ensure that the order adequately reflects the opinions of the participants.

In this paper, we propose an abstract model applicable to all these applications, and initiate a formal axiomatic study of the problem of finding a proportional collective ranking. In our study, we use tools from the political and social sciences, and in particular we adapt the concept of \emph{proportional representation} \citep{BaYo82a,Monr95a} 
%, common in the analysis of multiwinner elections,
to the case of rankings. Informally, proportional representation requires that the extent to which a particular preference or opinion is represented in the outcome should be proportional to the frequency with which this preference or opinion occurs within the population. For instance, proportional representation  is often a requirement in the context of \emph{parliamentary elections} \citep{Puke14a}, where candidates are grouped into political parties and voters express preferences over parties. If a party receives, say, 20\% of the votes, then proportional representation  requires that this party should be allocated (roughly) 20\% of the parliamentary seats \citep[see][for a discussion of arguments for and against proportional representation in a political context]{Lasl12a}.

The concept of proportional representation can be extended to the case of rankings in a very natural way. Intuitively, we say that a ranking $\tau$ is proportional if each prefix of $\tau$, viewed as a subset of alternatives, satisfies (some form of) proportional representation. We will require that, for each ``sufficiently large'' group of voters with consistent preferences, a proportional number of alternatives approved by this group is ranked appropriately high in the ranking. The position of such alternatives in the ranking depends on the level of their support, indicated by the size and the cohesiveness of the corresponding group of voters.

\paragraph{Our Contribution.}
The contribution of this paper is as follows:
\begin{inparaenum}[(i)]
\item We formalize the concept of proportionality of a ranking and introduce a quantitative way measuring it,
\item we observe that several known multiwinner rules satisfying \emph{committee monotonicity} can be viewed as rankings rules, 
\item we provide theoretical bounds on the proportional representation of several ranking rules, and
\item we experimentally evaluate ranking rules with respect to our measures of proportionality.
\end{inparaenum}

\paragraph{Related Work}

Proportional Representation is traditionally studied in settings where a \emph{subset} of alternatives (such as a parliament or a committee) needs to be chosen. The setting that is most often studied is that of \emph{closed party lists}, where alternatives are grouped into pairwise disjoint parties and voters are restricted to select a single party \citep{Gall91a}. In this setting, providing proportional representation reduces to solving an \emph{apportionment} problem \citep{BaYo82a,Puke14a}. 
In an influential paper, \citet{Monr95a} generalized the concept of proportional representation to settings where voter preferences are given as rank-orderings over the set of all alternatives. 
For approval preferences, concepts capturing proportional representation have recently been introduced by \citet{ABC+15a} and \citet{SFFB16a}.

The setting considered in this paper differs from all of the above settings in that we are interested in \emph{ranking} the alternatives, rather than choosing a subset of them. To the best of our knowledge, proportional rankings based on approval preferences have not been considered in the literature. In the context of linear (\ie rank-order) preferences, proportional rankings are discussed by \citet{Schu11b}. However, this paper neither proposes a measure of proportionality, nor does it compare the proportionality provided by different rules.  

%%%%%%%%%%%%%%%%%%%%%%%%%%%%%%%%%%%%%%%%%%%%%%%%%%%%%%

\section{Preliminaries}
\label{sec:prelims}
For $s\in{\mathbb N}$, we write $[s] = \{1, \dots, s\}$. For each set $X$, we let $\mathcal{S}(X)$ and $\mathcal{S}_k(X)$ denote the set of all subsets of $X$, and the set of all $k$-element subsets of $X$, respectively.

Let $N=[n]$ be a finite set of {\em voters} and $A = \{c_1, \ldots c_m\}$ a finite set of $m$ {\em alternatives}. 
Each voter $i\in N$ approves a non-empty subset of alternatives 
$A_i\subseteq A$. For each $a\in A$, we write $N_a$ 
for the set of voters who approve $a$, i.e., $N_a=\{i\in N: a \in A_i\}$. 
We refer to $|N_a|$ as the \emph{approval score} of $a$. 
A list $P = (A_1,\ldots, A_n)$ of approval sets, one for each voter $i \in N$, is called a \emph{profile on $(A,N)$}.
% We set $\mathcal{P}(A, N)$ to denote the set of all profiles on $(A,N)$.

A {\em ranking} is a linear order over~$A$.
For a ranking $r$ and for $k\in[m]$, we denote the
$k$-th element in $r$ by $r_k$. Thus, $r$ can be represented by the list $(r_1, r_2, \ldots, r_m)$.
Given a ranking $r = (r_1, \dots, r_m)$ and $k\in[m]$, 
we let $r_{\le k}=\{r_1, \dots, r_k\}$ denote the subset of $A$ consisting of the top $k$ elements according to $r$. 
An {\em (approval-based) ranking rule} $f$ maps a profile $P$ on $(A,N)$ to a ranking $f(P)$ over $A$. 

In what follows, we consider several ranking rules that can be obtained
by adapting existing multiwinner rules.
An \emph{(approval-based) multiwinner rule} takes as input a profile $P$ on $(A,N)$ and an integer $k \in [m]$ and outputs a $k$-element subset of $A$, referred to as the \emph{winning committee}.
A generic adaptation of multiwinner rules to ranking rules is possible whenever the respective 
multiwinner rule $\mathcal{R}$ has the property that $\mathcal{R}(P, k-1)\subseteq \mathcal{R}(P,k)$
for all $k\le m$ (this property is known as {\em committee monotonicity} or {\em house monotonicity}): whenever this is the case, 
we can produce a ranking by placing the unique alternative in 
$\mathcal{R}(P, k)\setminus \mathcal{R}(P, k-1)$ in position $k$. %of the ranking. 

\medskip
\noindent
We consider the following ranking rules.\footnote{All rules that we describe may have to break ties at some point in the execution;
we adopt an adversarial approach to tie-breaking, i.e., we say that a ranking rule
satisfies a property only if it satisfies it for all possible ways of breaking~ties.} 
Fix a profile $P$ on $(A,N)$. 
\begin{description}
\item[Approval Voting (AV).] Approval Voting ranks the alternatives in order of their approval score, so that $|N_{r_1}| \ge \dots \ge |N_{r_m}|$.

\item[Reweighted Approval Voting (RAV).] This family of rules is based on ideas developed by Danish polymath Thorvald Thiele \citep{Thie95a}. It is parameterized by {\em weight vectors}, \ie sequences $(w_1, w_2, \dots)$ of nonnegative reals. 
For a given weight vector $\vecw=(w_1, w_2, \dots)$ and a subset $S \subseteq A$ of alternatives, define
the {\em $\vecw$-RAV score of $S$} as $w(S) = \sum_{i \in N} \sum_{j=1}^{|A_i \cap S|} w_j$. The intuition behind this score is that voters prefer sets $S$ that contain more of their approved alternatives, but that there are decreasing marginal returns to adding further such alternatives to $S$ (if $\vecw$ is decreasing).
The rule $\vecw$-RAV constructs a ranking iteratively, starting 
with the empty partial ranking $r=()$. In step $k \in [m]$,
it appends to $r$ an alternative $a$ with maximum {\em marginal contribution} 
$w(r_{\le k-1}\cup \{a\}) - w(r_{\le k-1})$ among all yet unranked alternatives.
Many interesting rules belong to this family for suitable 
weight vectors $\vecw$. For example, AV is simply $(1,1,1,\ldots)$-RAV, 
\emph{Sequential Proportional Approval Voting (SeqPAV)} is defined as $\vecw_{\text{PAV}}$-RAV, 
where $\vecw_{\text{PAV}} = (1,\nicefrac{1}2,\nicefrac{1}3,\ldots)$,
\emph{Greedy Chamberlin--Courant} is defined as 
$(1,0,0,\ldots)$-RAV, and for every $p > 1$, 
the \emph{$p$-geometric rule} is given by $(\nicefrac{1}{p},\nicefrac{1}{p^2},\nicefrac{1}{p^3},\ldots)$-RAV. 

\item[Reverse SeqPAV.] This rule is a bottom-up variant of 
SeqPAV. It builds a ranking starting with the lowest-ranked alternative. 
Initially it sets $S=A$ and $r=()$. Then, 
at each step it picks an alternative $a$ minimizing 
$w_{\text{PAV}}(S) - w_{\text{PAV}}(S\setminus \{a\})$, 
removes it from $S$ and prepends it to $r$.

\item[Phragm\'{e}n's Rule.] \citet{Phra95a} proposed a committee selection rule based on a load balancing approach: every alternative incurs a \emph{load} of one unit, and the load of alternative~$a$ has to be distributed among all voters in $N_a$. Phragm\'{e}n's rule constructs a ranking iteratively, starting with the empty partial ranking $r=()$. Initially, the load of each voter is~$0$. At each step, the rule picks a yet unranked alternative and distributes its associated load of $1$ over the voters who approve it; the alternative and the load distribution scheme are chosen so as to minimize the maximum load across all voters. This alternative is then appended to the ranking $r$. (For details, see the work of \citealp{Jans12a}, and \citealp{MoOl15a}).
\end{description}

The following example illustrates the ranking rules defined above.

\begin{figure}[tb]
	\centering
\begin{tikzpicture}[yscale=.45,xscale=.1]
	\draw (0,7) rectangle node {$c_1$} +(16.7,1); \draw (16.7,7) rectangle node {$c_2$} +(20.0,1); \draw (36.7,7) rectangle node {$c_3$} +(20.0,1); \draw (56.7,7) rectangle node {$c_4$} +(20.0,1);
	\draw (0,6) rectangle node {$c_1$} +(16.7,1); \draw (16.7,6) rectangle node {$c_2$} +(20.0,1); \draw (36.7,6) rectangle node {$c_3$} +(20.0,1); \draw (56.7,6) rectangle node {$c_4$} +(20.0,1);
	\draw (0,5) rectangle node {$c_1$} +(16.7,1); \draw (16.7,5) rectangle node {$c_2$} +(20.0,1); \draw (36.7,5) rectangle node {$c_3$} +(20.0,1); \draw (56.7,5) rectangle node {$c_4$} +(20.0,1);
	\draw (0,4) rectangle node {$c_6$} +(58.3,1); \draw (58.3,4) rectangle node {$c_5$} +(45.6,1);
	\draw (0,3) rectangle node {$c_1$} +(16.7,1); \draw (16.7,3) rectangle node {$c_2$} +(20.0,1); \draw (36.7,3) rectangle node {$c_3$} +(20.0,1); \draw (56.7,3) rectangle node {$c_4$} +(20.0,1); \draw (76.7,3) rectangle node {$c_5$} +(27.2,1);
	\draw (0,2) rectangle node {$c_1$} +(16.7,1); \draw (16.7,2) rectangle node {$c_2$} +(20.0,1);  \draw (36.7,2) rectangle node {$c_3$} +(20.0,1); \draw (56.7,2) rectangle node {$c_4$} +(20.0,1); \draw (76.7,2) rectangle node {$c_5$} +(27.2,1);
	\draw (0,1) rectangle node {$c_1$} +(16.7,1); \draw (16.7,1) rectangle node {$c_6$} +(41.6,1);

	\draw (16.7,8.8) node[]{$\frac{1}{6}$}; \draw (36.7,8.8) node[]{$\frac{11}{30}$}; \draw (56.7,8.8) node[]{$\frac{17}{30}$}; \draw (76.7,8.8) node[]{$\frac{23}{30}$};
	
	\draw (103.9,5.8) node[]{$\frac{187}{180}$};
	
	\draw (58.3,0.2) node[]{$\frac{7}{12}$};
	
	\draw (-0.5,7.5) node[left]{$A_1 = \{c_1,c_2,c_3,c_4\}$};
	\draw (-0.5,6.5) node[left]{$A_2 = \{c_1,c_2,c_3,c_4\}$};
	\draw (-0.5,5.5) node[left]{$A_3 = \{c_1,c_2,c_3,c_4\}$};
	\draw (-0.5,4.5) node[left]{$A_4 = \{c_5,c_6\}$};
	\draw (-0.5,3.5) node[left]{$A_5 = \{c_1,c_2,c_3,c_4,c_5 \}$};
	\draw (-0.5,2.5) node[left]{$A_6 = \{c_1,c_2,c_3,c_4,c_5 \}$};
	\draw (-0.5,1.5) node[left]{$A_7 = \{c_1,c_6 \}$};
\end{tikzpicture}
\caption{Illustration of the load distribution produced by Phragm\'{e}n's rule in \exref{ex:owaBasedRules}.}
\label{fig:ex-phragmen}
\end{figure}

\begin{example}\label{ex:owaBasedRules}
{\em
Consider the following profile with 7 voters and 6 alternatives:
\begin{align*}
& A_1 = A_2 = A_3 = \{c_1, c_2, c_3, c_4\} \quad & & A_4 = \{c_5, c_6\} \\
& A_5 = A_6 = \{c_1, c_2, c_3, c_4, c_5\} \quad & & A_7 = \{c_1, c_6\}
\end{align*}
% \[ A_1 = A_2 = A_3 = \{c_1, c_2, c_3, c_4\}, \;A_4 = \{c_5, c_6\}, \; A_5 = A_6 = \{c_1, c_2, c_3, c_4, c_5\}, \; A_7 = \{c_1, c_6\} \]
Assume lexicographic tie-breaking. Since the approval scores of the alternatives $c_1, \ldots, c_6$ are equal to 6, 5, 5, 5, 3, and 2, respectively,
Approval Voting returns the ranking $(c_1, c_2, c_3, c_4, c_5, c_6)$. 

SeqPAV in the first iteration selects an alternative with the highest approval score, i.e., $c_1$. In the second iteration, the marginal contribution of alternatives $c_2, \ldots, c_6$ to the PAV-score are equal to $\nicefrac{5}{2}$, $\nicefrac{5}{2}$, $\nicefrac{5}{2}$, $2$, and $\nicefrac{3}{2}$, respectively. Thus, SeqPAV selects $c_2$. In the next iterations, the rule appends $c_3$, $c_5$, $c_4$, and finally $c_6$ to the ranking.

Let us now consider Reverse SeqPAV. Removing alternative $c_1$ from $\{c_1, \ldots, c_6\}$ decreases the PAV-score of the voters by $3\cdot \frac{1}{4} + 2\cdot \frac{1}{5} + \frac{1}{2} = 1.65$. Removing alternatives $c_2, \ldots, c_6$ from $\{c_1, \ldots, c_6\}$ decreases the PAV-score of the voters by $1.5$, $1.5$, $1.5$, $0.9$ and $1.0$, repectively. Thus, alternative $c_5$ is put in the last position of the ranking produced by Reverse SeqPAV. The whole ranking returned by this rule is $(c_1, c_2, c_3, c_6, c_4, c_5)$.

The reader can easily verify that the $2$-geometric rule and the Greedy Chamberlin--Courant rule return the rankings $(c_1, c_5, c_2, c_3, c_4, c_6)$ and $(c_1, c_2, c_5, c_3, c_6, c_4)$, respectively.

Finally, Phragm\'{e}n's rule returns the ranking $(c_1,c_2,c_3,c_6,c_4,c_5)$. The corresponding load distribution is illustrated in \figref{fig:ex-phragmen}.
}
\end{example}

\section{Measures of Proportionality}

\newcommand{\gr}{\cal{GR}\xspace}
\newcommand{\agr}{$\alpha$-\cal{GR}\xspace}
\newcommand{\akgr}{$\alpha(k,\ell)$-\cal{GR}\xspace}

\newcommand{\sigi}{significant\xspace}
\newcommand{\avg}{\mathrm{avg}}
\newcommand{\jd}{\mathrm{jd}}
\newcommand{\jar}{\jd}

\newcommand{\pro}{\alpha}
\newcommand{\coh}{\lambda}

In this section, we define a measure of proportionality for rankings
and then extend it to ranking rules.
In what follows, let $P$ be a profile on $(A,N)$ with $|A|=m$. 

Given a group of voters $N'\subseteq N$ and a set of alternatives $S\subseteq A$,
a natural measure of the group's ``satisfaction'' provided by $S$ is the average number
of alternatives in $S$ that are approved by a voter in $N'$. Thus, we define
\[ 
\avg(N',S) = \frac{1}{|N'|} \sum_{i \in N'} |A_i \cap S| \text.
\]
We refer to $\avg(N',S)$ as the \emph{average representation of $N'$ with respect to $S$}.
To extend this idea to rankings, we consider
the case where the subset $S$ is an initial segment of a given ranking $r$,
\ie $S=r_{\le k}$ for some $k\in [m]$. 
Intuitively, every 
group $N'$ wants to have an average representation $\avg(N', r_{\le k})$ 
that is as large as possible, for all $k\in [m]$.
Now, whether a group of voters deserves to be represented in the top positions 
of a ranking depends on two parameters: 
its relative size and its cohesiveness, i.e., 
the number of alternatives that are unanimously approved by the group members. 
This motivates the following definition.

\begin{definition}{\bf (Significant groups)}
Consider a profile $P$ on $(A,N)$.
For a group $N'\subseteq N$, its \emph{proportion $\pro(N')$} is given by 
$\pro(N') = \frac{|N'|}{|N|}$, 
and its \emph{cohesiveness $\coh(N')$} is given by 
$\coh(N') = |\bigcap_{i \in N'} A_i|$.
Given $\alpha\in(0, 1]$ and $\lambda\in[m]$,
we say that $N'$ is {\em $(\alpha, \lambda)$-significant in $P$} if %$\pro(N')\ge \alpha$ 
$|N'|=\lceil\alpha n\rceil$ and $\coh(N')\ge \lambda$.
\end{definition}

The following definition captures the intuitively compelling idea that a group can demand to be proportionally 
represented in the top positions of the ranking, as long as the number of demanded alternatives 
does not exceed the cohesiveness of the group.

\begin{definition} {\bf (Justifiable demand)}
The \emph{justifiable demand} of a group $N'\subseteq N$ with respect to the top $k$ positions of a ranking is defined as 
\[ 
\jar(N',k)=\min( \lfloor \pro(N') \cdot k \rfloor, \coh(N')) \text. 
\]
\end{definition}

For example, if a group contains 25\% of the voters and has a cohesiveness of 3, it has a justifiable demand of 1 
with respect to the top four positions, a justifiable demand of 2 with respect to the top eight positions, and a 
justifiable demand of 3 with respect to the top twelve positions, which is also its maximum justifiable demand.

It would be desirable to find a rule that provides every group with an average representation that meets the group's justifiable demand. However, the following example shows that this is not always possible.

\begin{example}\label{ex:lower_bound}
{\em
Let $A = \{a,b,c\}$ and $n=6$ and consider the profile $P$ given by 
$A_1=\{a\}$, $A_2=\{a,b\}$, $A_3=\{b\}$, $A_4=\{b,c\}$, $A_5=\{c\}$, and $A_6=\{a,c\}$. 
Consider the ranking $r=(a,b,c)$. The group $N'=\{4,5,6\}$ has $\pro(N')=\nicefrac{1}{2}$, $\coh(N')=1$. 
Therefore, its justifiable demand with respect to the top two positions is 
$\jd(N',2) = \min(\lfloor \nicefrac{1}{2} \cdot 2 \rfloor, 1)=1$.
However, its average representation with respect to the top two positions 
of $r$ is only  $\avg(N',\{a,b\})=\nicefrac{2}{3}$.
Since $P$ is completely symmetric with respect to the alternatives, 
we can find such a group for every other ranking as well. 
\label{ex:1/2isbest}
}
\end{example}

This example shows that it may not be feasible to provide each group 
with the level of representation that meets its justifiable 
demand, but it might be possible to guarantee a large fraction of it. For instance, in the previous example it is 
possible to ensure that $\avg(N',k)\ge \nicefrac{2}{3}\cdot \jar(N',k)$ for all groups $N'$ and for all $k\le 3$.
This observation leads to the following definition.

\begin{definition}\label{def:opt} {\bf (Optimal ranking)}
We define the {\em quality} of a ranking $r$ for a profile $P$ as
\[ 
q_P(r) = \min_{\substack{k\in [m], N'\subseteq N:\\ \jar(N', k)>0}} \frac{\avg(N', r_{\le k})}{\jar(N',k)} \; \text.
\]
An {\em optimal ranking for $P$} is a ranking in $\arg\max_{r}q_P(r)$.
\end{definition}
%EE: can we claim q_P(r)\le 1 always?
We observe that $q_P(r)$ is well-defined, i.e., the set of pairs $(N', k)$
such that $\jar(N', k)>0$ is always non-empty. Indeed, consider an alternative 
$a$ with the highest approval score: as $A_i\neq\emptyset$ for all $i\in N$,
by the pigeonhole principle we have $|N_a|\ge\nicefrac{n}{m}$. Moreover, $\coh(N_a)\ge 1$,
so $\jar(N_a, m)\ge\min(\frac{\nicefrac{n}{m}}{n}\cdot m, 1)=1$.
Example~\ref{ex:1/2isbest} illustrates
that the quality of an optimal ranking may be less than $1$. 

%EE: one can think of q_P(r) as a measure of cohesiveness of P, a la Hashemi-Endriss
Unfortunately, good rankings are hard to compute.

\begin{theorem}
Given a profile $P$, it is $\np$-hard to decide whether there exists a ranking $r$ with
$q_P(r) \ge 1$. \label{thm:optimal-hard}
% \[\avg(N',\{r_1,\dots,r_k\})\geq {\jar(N',k)}\] for all $N'\subseteq N$ and $1\leq k \leq m$.
\end{theorem}
\begin{proof}
We give a reduction from \textsc{Vertex Cover}. An instance of this problem is given by 
a graph $G=(V, E)$ and an integer $\ell$.
It is a `yes'-instance if there is a subset of $\ell$ vertices $S \subseteq V$ 
such that each edge in $E$ contains some vertex in $S$. We can assume that $G$
is $3$-regular, as
{\sc Vertex Cover} remains $\np$-hard in this case \citep{gar-joh:b:int}.

Given a vertex cover instance $(G, \ell)$ with $G=(V, E)$, we construct 
an instance of the problem of finding an optimal ranking in the following way. 
Since the degree of each vertex is exactly 3, 
we can assume that $3\ell \geq |E|$: 
instances with $3\ell < |E|$ are trivially `no'-instances.
We set $A=V\cup \{d\}$, where $d$ is a dummy alternative. 
For every edge $\{a,b\}\in E$ we create a voter that approves $\{a,b\}$. 
We also add $3\ell+3-|E|\ge 3$ dummy voters who approve $d$ only.
Consequently, 
$n = |N| = 3\ell+3$. Since the graph is $3$-regular,
we have $m=|A|= 2|E|/3+1 \leq 3\ell$. This defines the profile $P$.
We will show that $G$ has a vertex cover of size $\ell$ if and only if 
there exists a ranking $r$ with $q_P(r)\ge 1$, \ie
$\avg(N', r_{\le k})\geq {\jar(N',k)}$ for all $N'\subseteq N$ and $k\in[m]$.

$\Rightarrow$: Let $S\subseteq V$ be a vertex cover of size $\ell$. 
Consider a ranking where $d$ is ranked first, the next $\ell$ positions are occupied by elements of $S$ 
(in arbitrary order), followed by the remaining alternatives (also in arbitrary order).
Consider a group of voters $N'\subseteq N$ and an integer $k\in[m]$; 
if $\coh(N')=0$ then trivially $\avg(N', r_{\le k})\geq {\jar(N',k)}=0$,
so assume that $\coh(N')>0$. 
We consider the following cases:
\begin{itemize}
\item 
$N'\subseteq N_d$. Then $\coh(N')=1$, so for each $k\in[m]$
we have $\avg(N', r_{\le k})\ge \avg(N', r_{\le 1})=1=\coh(N')\ge\jar(N', k)$.
\item
$N'\cap N_d=\emptyset$, $|N'|>1$.
Since the intersection of two or more edges contains at most one vertex 
and the degree of each vertex is $3$, we have $\coh(N')=1$ and $|N'|\le 3$.
For $k\le \ell$, we obtain $\jar(N',k)\leq\min(\lfloor\frac{3}{3\ell+3}\cdot k\rfloor,1)=0$; 
for $k\ge \ell+1$, we have $\jar(N',k')\le 1$.
Since $r_{\le \ell+1}= S\cup \{d\}$ and $S$ is a vertex cover for $G$, 
it holds that $\avg(N', r_{\le k})\geq \avg(N', r_{\le \ell+1})\geq 1 \geq \jar(N',k)$.
\item
$N'\cap N_d=\emptyset$, $|N'|=1$, i.e.,
$N'$ contains a single edge voter.
Then $\jar(N', k)\le \pro(N', k)\le \frac{k}{3\ell+3}<1$ for all $k\le m$ (recall that $m\le 3\ell$),
so $\jar(N', k)=0$ for all $k\in [m]$.
\end{itemize}
Hence $\avg(N', r_{\le k})\geq {\jar(N',k)}$ for all $N'\subseteq N$, $k\in [m]$.

$\Leftarrow$: Consider a ranking $r$ with
$\avg(N', r_{\le k})\geq {\jar(N',k)}$ for all $N'\subseteq N$ and $k\in [m]$.
We claim that $d\in r_{\le\ell+1}$ and $S= r_{\le \ell+1}\setminus\{d\}$ is a vertex cover.
Note that we have $N_d\ge 3$, $\coh(N_d)=1$, 
so $\jar(N_d, \ell+1)\ge \min(\frac{3}{3\ell+3}(\ell+1), 1)=1$.
As voters in $N_d$ only approve $d$, 
alternative $d$ has to appear in the first $\ell+1$ positions of the ranking. 
Similarly, for every alternative $a\in A$ we have $N_a=3$, so
$\pro(N_a)=\frac{3}{3\ell+3}$ and $\coh(N_a)=1$.
For $k=\ell+1$, the justifiable demand of $N_a$ is $\jar(N_a,k)=\min(\frac{3(\ell+1)}{3\ell+3},1)=1$.
Now suppose that there exists an edge $e=\{a,b\}$ with $e\cap S = \emptyset$.
Let $N_a=\{e, e', e''\}$; we have $|e'\cap S|\le 1$, $|e''\cap S|\le 1$, 
so $\avg(N_a,S)\leq \nicefrac{2}{3}$, a contradiction.
We conclude that $S= r_{\le \ell+1}\setminus\{d\}$ is a vertex cover.
\end{proof}

Further, a fairly straightforward reduction from the Maximum $k$-Subset Intersection problem~\citep{Xavier2012471} shows that even the problem of deciding whether there exists an $(\alpha, \lambda)$-significant group of voters is $\np$-hard.

\begin{proposition}\label{prop:significantGroupsAreHard}
The problem of deciding whether there exists an $(\alpha, \lambda)$-significant group of voters is $\np$-complete.
\end{proposition}

\begin{proof}
It is clear that the problem belongs to $\np$. Below we show that it is $\np$-hard.

Let $I$ be an instance of the Maximum k-Subset Intersection problem. In $I$ we are given a collection $\mathcal{S} = \{S_1, \ldots, S_{m}\}$ of $m$ subsets over a set of $n$ elements $\mathcal{E} = \{e_1, \ldots, e_n\}$, and two positive integers, $\lambda$ and $\beta$. We ask whether there exists a subcollection of $\lambda$ sets from $\mathcal{S}$, $\{S_{j_1}, \ldots, S_{j_{\lambda}}\}$, such that $|S_{j_1} \cap \ldots \cap S_{j_{\lambda}}| \geq \beta$.

From $I$ we can construct an instance of the problem of deciding whether there exists an $(\alpha, \lambda)$-significant group of voters, in the following way. We associate sets from $\mathcal{S}$ with alternatives, and the elements from $\mathcal{E}$ with voters (a voter $e_i$ approves of an alternative $S_j$ if $e_i \in S_j$). We set $\alpha = \frac{\beta}{|\mathcal{S}|}$. It is easy to see that for each $\lambda$ alternatives $S_{j_1}, \ldots, S_{j_{\lambda}}$, $|S_{j_1} \cap \ldots \cap S_{j_{\lambda}}|$ is the size of the largest group of voters who all approve the $\lambda$ alternatives, and that there exists an $(\alpha, \lambda)$-significant group of voters if and only if there exist $\lambda$ alternatives that are all approved by some $\beta$ voters.  
\end{proof}

%\noindent \textbf{Open problem:} Does an optimal ranking $r$ guarantee $q_P(r) \ge 2/3$ for all $P$?

The measure $q_P(r)$ can be lifted from individual rankings to ranking rules:
we can measure the quality of a ranking rule $f$ as the minimum value of $q_P(f(P))$,
over all possible profiles~$P$. However, in our theoretical analysis of ranking rules
we take the following approach, which assumes more flexibility on the part of groups
of voters that seek to be represented:
given a group of proportion $\alpha$ and with cohesiveness at least $\lambda$, 
we ask at which point in the ranking the average satisfaction of this group reaches $\lambda$.
This guarantee is given by the function $\kappa(\alpha,\lambda)$ and formally defined as follows.

\begin{definition} {\bf ($\kappa$-group representation)}
Let $\kappa(\alpha, \lambda)$ be a function from $((0,1]\cap{\mathbb Q} )\times \mathbb{N}$ to $\mathbb{N}$.
A ranking $r$ provides \emph{$\kappa$-group representation ($\kappa$-GR) for profile $P$} if
for all rational $\alpha\in(0,1]$, all $\lambda\in \mathbb{N}$, and all voter groups $N'\subseteq N$ 
that are $(\alpha,\lambda)$-\sigi in $P$ it holds that 
\[
\avg(N', r_{\le \kappa(\alpha,\lambda)}) \ge \lambda.
\]
A ranking rule $f$ satisfies \emph{$\kappa$-group representation ($\kappa$-GR)} if 
$f(P)$ provides $\kappa$-group representation for every profile $P$. 
\end{definition}

Let us now explore the differences between the two measures, the worst-case quality of the ranking $q_P(f(P))$, and $\kappa$-group representation. While $q_P(f(P))$ is just a single number, $\kappa$-group representation carries more information. In particular, it makes it possible to express the fact that some groups are better represented in the top parts of the resulting ranking than further below (as would be suggested by $\kappa$ being a convex function over $\lambda$). As a more concrete example, assume that $q_P(f(P)) < \nicefrac{1}{2}$ and consider a group $N'$ that is $(\nicefrac{1}{10}, \nicefrac{m}{2})$-significant. Since the justified demand of this group is upper bounded by $\nicefrac{m}{10}$, $q_P(f(P))$ does not say anything about how far we need to go down the ranking to obtain an average representation greater than $\lambda > \nicefrac{m}{20}$ 
(the justifiable demand of $N'$ is at most $\nicefrac{m}{10}$ % MB: why? 
and $f$ guarantees only half of it), 
while $\kappa$-group representation provides such information for each $\lambda \in [1, \nicefrac{m}{2}]$.

In the next section, we investigate guarantees in terms of $\kappa$-group representation provided
by the ranking rules that we introduced above.

%%%%%%%%%%%%%%%%%%%%%%%%%%%%%%%%%%%%%%%%%%%%%%%%%%%%%%%%%%%%%%%%%%%%%%%%%%%%%%%%%%%

\section{Theoretical Guarantees of Ranking Rules}\label{sec:worstcase}
We start our analysis with the simplest of our rules, namely, Approval Voting. 
Interestingly, Approval Voting provides very good guarantees to majorities, that is,
groups with $\pro(N')\ge \nicefrac12$. However, for minorities it provides no guarantees at all. 

\begin{theorem}\label{thm:av}
For rational $\alpha > \nicefrac12$, Approval Voting satisfies
$\kappa(\alpha, \lambda)$-group representation for $\kappa(\alpha, \lambda) =
\big\lceil \frac{\lambda\alpha}{2\alpha - 1} \big\rceil$,
but fails it for $\kappa(\alpha, \lambda) =
\big\lceil\frac{\lambda\alpha}{2\alpha - 1}\big\rceil - 1$.
For $\alpha < \nicefrac12$, Approval Voting does not satisfy
$\kappa(\alpha, \lambda)$-group representation for any function $\kappa(\alpha, \lambda)$.
\end{theorem}
\begin{proof}
We first prove that for $\alpha > \nicefrac{1}{2}$ 
Approval Voting satisfies $\kappa(\alpha, \lambda)$-group representation 
for $\kappa(\alpha, \lambda) = \lceil\frac{\lambda\alpha}{2\alpha - 1}\rceil$.
Consider a group of voters $N' \subseteq N$ with $n'=|N'| = \lceil\alpha n\rceil$, $\coh(N')\ge \lambda$.
Let $k = \big\lceil \frac{\lambda\alpha}{2\alpha - 1} \big\rceil$,
and let $r$ be the ranking returned by Approval Voting.
Assume for the sake of contradiction that $\avg(N', r_{\le k}) < \lambda$.
Then there exists some alternative $h$ that is approved by all voters in $N'$,
but does not appear in the top $k$ positions of $r$. Thus, each alternative
in $r_{\le k}$ is approved by at least $\lceil \alpha n\rceil$ voters
and hence by at least $2\lceil \alpha n \rceil- n$ voters in $N'$. 
Hence, 
\begin{align*}
\avg(N', r_{\le k}) \ge
k \cdot \frac{2\lceil \alpha n \rceil - n}{\lceil \alpha n \rceil} \geq 
k\cdot \frac{2\alpha n - n}{\alpha n} \geq \lambda \text{.}
\end{align*}
This contradiction proves our first claim.

Now, we will show that Approval Voting does not satisfy $\kappa(\alpha, \lambda)$-group representation 
for $\kappa(\alpha, \lambda) = 
\big\lceil\frac{\lambda\alpha}{2\alpha - 1}\big\rceil - 1$. 
Fix $\alpha\in(0, 1]\cap{\mathbb Q}$ and $\lambda\in{\mathbb N}$ and
let $k = \big\lceil\frac{\lambda\alpha}{2\alpha - 1}\big\rceil - 1 < \frac{\lambda\alpha}{2\alpha - 1}$.
As $\alpha\in{\mathbb Q}$, there exist some integers $x$ and $n$ such that 
$\alpha=\nicefrac{x}{n}$. 
Let $A = A'\cup A''$, where $|A'|=|A''|=k$ and $A'\cap A''=\emptyset$.
Consider a profile on $(A,N)$ that contains two groups of voters $G$ and $F$ with sizes $|G| = |F| = x = \alpha n$,
such that $G$ and $F$ have the smallest possible intersection. That is, 
for $\alpha \leq \nicefrac{1}{2}$ the sets $G$ and $F$ are disjoint, 
and for $\alpha > \nicefrac{1}{2}$ we have $|G \cap F| = (2\alpha - 1)n$. 
Suppose that each voter in $G$ approves all alternatives in $A'$ and 
each voter in $F$ approves all alternatives in~$A''$. 
Approval Voting may rank $A'$ in the first $k$ positions.
Thus, 
\begin{align*}
\avg(F, r_{\le k})  = k \cdot \frac{(2\alpha - 1)n}{\alpha n} < \lambda \text{.}
\end{align*}

Finally, let us prove that for $\alpha \leq \nicefrac{1}{2}$ Approval Voting does not satisfy $\kappa(\alpha, \lambda)$-group
representation for any function $\kappa(\alpha, \lambda)$. 
For the sake of contradiction let us assume that this is not the case.
Let us fix $\alpha\in(0, 1]$, $\lambda\in{\mathbb N}$, 
and let $k = \kappa(\alpha, \lambda)$. 
Using the same idea as in the previous paragraph, 
we obtain an instance where there is a set of voters $N'$ with $|N'|=\alpha n$, $\coh(N')=|N'|\ge \lambda$ 
such that no voter in $N'$ approves any of the alternatives appearing
in top $k$ positions of the ranking returned by Approval Voting. 
This gives a contradiction and proves our claim.
\end{proof}

In contrast, Phragm\'{e}n's rule, SeqPAV and $p$-geometric rules,
provide reasonable guarantees for \emph{all} values of $\alpha$ and $\lambda$.

\begin{theorem}\label{thm:phragmen}
Phragm\'{e}n's rule satisfies $\kappa(\alpha, \lambda)$-group representation for 
$\kappa(\alpha, \lambda) = \lceil\frac{5\lambda}{\alpha^2} + \frac{1}{\alpha}\rceil$.
\end{theorem}
\begin{proof}
Fix $\alpha\in (0, 1]$, $\lambda\in{\mathbb N}$, a profile $P$, 
and a group of voters $N' \subseteq N$ such that $n' = |N'| = \lceil \alpha n \rceil$
and $\coh(N')\ge \lambda$. Set $y = \nicefrac{1}{n'}$. 
Let $\ell_{i}(t)$ denote the load of voter $i$ after the $t$-th step of Phragm\'{e}n's rule. 
Let $\mu(t) = \max_{i \in N'} \ell_i(t)$, \ie
$\mu(t)$ is the maximum load across all voters from $N'$ after the $t$-th step.
Further, define the {\em excess of voter $i$ at step $t$} 
as $e_{i}(t) = \mu(t) - \ell_i(t)$, set $\delta_i(t) = \ell_i(t) - \ell_i(t-1)$, and let
\begin{align*}
E(t) = \sum_{i \in N'} e_{i}(t) \quad \text{and} \quad S(t) = \sum_{i \in N'} e_{i}(t)^2 \text{.}
\end{align*}
The above definitions are illustrated in Figure~\ref{fig:phragmen_notation}.

Let $k = \lceil\frac{5\lambda}{\alpha^2} + \frac{1}{\alpha}\rceil$ and let $r$ be ranking returned by Phragm\'{e}n's rule. 
For the sake of contradiction, let us assume that $\avg(N',r_{\le k}) < \lambda$. 
Then during the first $k$ steps of the rule 
there exists an alternative that is approved 
by all members of $N'$, but has not yet been selected by the rule;
denote this alternative by $h$.

First, we prove that for all $t\in[k]$ it holds that $E(t) \leq 1$. 
In fact, we will prove a stronger claim:
\begin{align}\label{ineq:boundedGapArea}
\sum_{i \in N'}\Big(\max_{i \in N}\ell_i(t) - \ell_i(t)\Big) \leq 1 \text{.}
\end{align}
(The claim is stronger because we take the maximum over all voters rather 
than over the voters from $N'$.)  
Indeed, suppose for the sake of contradiction that this is not the case; 
let $t$ be the first step of the algorithm where Inequality~\eqref{ineq:boundedGapArea}
is violated. %let $a$ be the alternative that was selected at that step. 
Since the loads of all voters monotonically increase, 
we have $M = \max_{i \in N}\ell_i(t) > \max_{i \in N}\ell_i(t-1)$,
\ie the maximum load strictly increases. However, since $E(t)>1$, 
the algorithm could have selected $h$ and distributed the associated load in such a way 
that the maximum load is lower than $M$, a contradiction. Thus, $E(t) \leq 1$ for each $t\in[k]$. 

\begin{figure}[tb]
\begin{center}
   \includegraphics[scale=0.5]{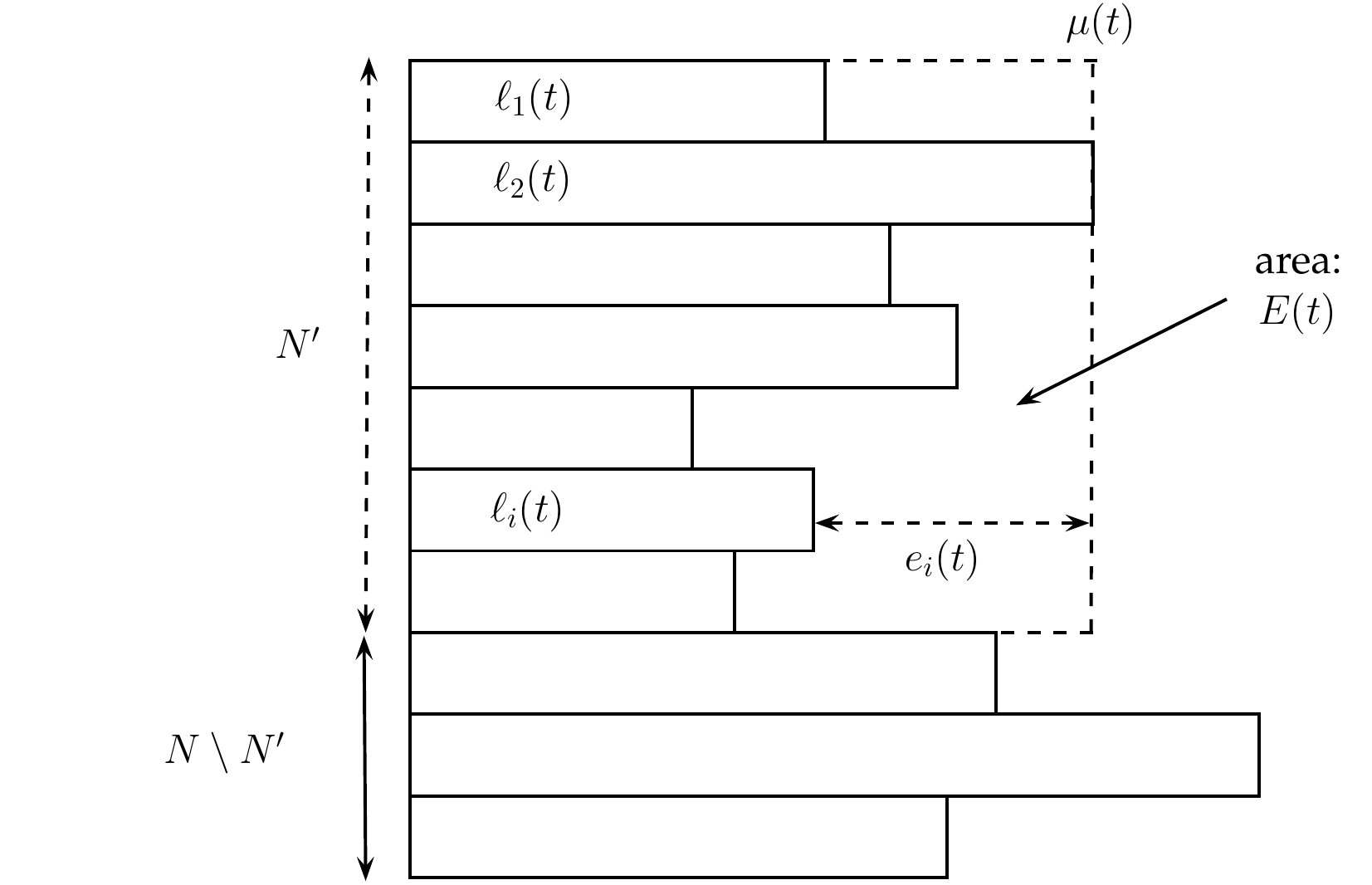}
\end{center}
\caption{Illustration of the notation used in the proof of Theorem~\ref{thm:phragmen}.}
\label{fig:phragmen_notation}
\end{figure}

Second, we prove that 
\begin{align}\label{ineq:sumOfLoads}
\sum_{i \in N'} \ell_i(k) \geq n'\Big(\frac{k}{n} - y\Big) \text{.}
\end{align}
Again, for the sake of contradiction 
assume that $\sum_{i \in N'} \ell_i(k) < n'\big(\frac{k}{n} - y\big) = n'\frac{k}{n}-1$. 
Note that the total load assigned to the voters after $k$ steps of the algorithm 
is equal to $k$. Thus, by the pigeonhole principle, after the $k$-th round
the load of some voter is at least $\frac{k}{n}$. 
Let $t$ be the first round where the load of some voter is at least $\frac{k}{n}$. 
We have
$$
\sum_{i \in N'} \ell_i(t)\le \sum_{i \in N'} \ell_i(k) < n'\frac{k}{n} - 1.
$$
Thus, the rule could have chosen $h$ and distributed the associated load
in such a way that the load of each voter is less than $\frac{k}{n}$, 
a contradiction. 
This proves Inequality~\eqref{ineq:sumOfLoads}. 
Since $n'(\frac{k}{n} - y) = n'\frac{k}{n}-1  \geq \alpha k - 1$, we infer that
\begin{equation}
\sum_{i \in N'}\sum_{t=1}^k \delta_{i}(t) = \sum_{i \in N'}\sum_{t=1}^k \Big( \ell_i(t) - \ell_i(t-1) \Big) 
= \sum_{i \in N'}\Big( \ell_i(k) - \ell_i(0) \Big) = \sum_{i \in N'}\ell_i(k) \geq \alpha k - 1.
\end{equation}

Further, let us investigate the relationship between the average representation of voters from $N'$ 
and the value $S(t)$. First, observe that $\mu(t+1) - \mu(t) \leq y$: indeed, if that was not the case,
the algorithm could have chosen $h$ at step $t$ and distributed 
the associated load uniformly among the voters in $N'$ to achieve a lower maximum load. 

% MB: simplified
Consider a voter $i \in N'$. Since $\mu(t+1) - \mu(t) \leq y$, it follows that $\delta_{i}(t+1) \le e_i(t) + y$.
% [old]
% Moreover, 
% \begin{align*}
% y + e_{i}(t) &= y + \mu(t) - \ell_i(t) \\
%              &\geq \mu(t+1) - \mu(t) + \mu(t) - \ell_i(t) \\
%              &= \Big(\mu(t+1)- \ell_i(t+1)\Big) + \Big(\ell_i(t+1) - \ell_i(t)\Big) \\
%              &= \Big(\mu(t+1)- \ell_i(t+1)\Big) + \delta_{i}(t+1) \geq \delta_{i}(t+1) \text{;}
% \end{align*}
Multiplying by $-2\delta_i(t+1)$, we obtain
\begin{align}\label{ineq:rvEstimation}
-2(y + e_{i}(t))\delta_{i}(t+1) + \delta_{i}^2(t+1) \leq -\delta_{i}^2(t+1) \text{.}
\end{align}
Also, for all $t\in[k-1]$ we have
\begin{equation}\label{ineq:gapExpansion}
	e_{i}(t+1) \leq e_{i}(t) + y - \delta_{i}(t+1) \text{.}
\end{equation}

% MB: simplified
% [old]
% \begin{align}\label{ineq:gapExpansion}
% \begin{split}
% e_{i}(t+1) &= \mu(t+1) - \ell_i(t+1) \\
%       &= \Big(\mu(t+1) - \mu(t)\Big) + \Big(\mu(t) - \ell_i(t)\Big) + \Big(\ell_i(t) - \ell_i(t+1)\Big)
%       \\&\leq y + e_{i}(t) - \delta_{i}(t+1) \text{.}
% \end{split}
% \end{align}
We are ready to assess the value $S(t+1) - S(t)$:
\begin{align}
S(t+1) - S(t) &= \sum_{i \in N'} \Big(e_{i}^2(t+1) - e_{i}^2(t)\Big) \\
              &\leq \sum_{i \in N'} \Big((e_i(t) + y - \delta_i(t+1))^2 - e_{i}^2(t)\Big) \label{ineq:gapExpansionUsage} \\
              &= \sum_{i \in N'} \Big(2e_i(t)y - 2(y + e_{i}(t))\delta_i(t+1) + y^2 + \delta_i^2(t+1)\Big) \label{ineq:gapReduction}\\
              &\leq  \sum_{i \in N'} \Big(2e_i(t)y - \delta_i^2(t+1) + y^2\Big) \label{ineq:rvEstimationUsage}\\
              &= 2yE(t) + y^2 n' - \sum_{i \in N'}\delta_i^2(t+1) \label{ineq:gapEstimationUsage} \\
              &\leq 3y - \sum_{i \in N'}\delta_{i}^2(t+1) \textrm{.} \label{ineq:last}
\end{align}
In the above sequence of inequalities, 
\eqref{ineq:gapExpansionUsage} follows from \eqref{ineq:gapExpansion}, 
\eqref{ineq:gapReduction} and \eqref{ineq:gapEstimationUsage} follow by simple algebraic operations, 
\eqref{ineq:rvEstimationUsage} is a consequence of \eqref{ineq:rvEstimation}, 
and \eqref{ineq:last} follows from the fact that $E(t) \leq 1$ and $yn' = 1$.
As a result, we get
\begin{align*}
0 \leq S(k) - S(0) = \sum_{t = 1}^{k} \Big(S(t) - S(t-1)\Big) \leq 3ky - \sum_{i \in N'}\sum_{t=1}^k\delta_i^2(t) \textrm{,}
\end{align*}
and thus $\sum_{i \in N'}\sum_{t=1}^k \delta_i^2(t) \leq 3ky$. 

To summarize, we obtained an upper bound of $3ky$ on the sum of squares of the variables 
from $\{\delta_i(t)\}_{i \in N', t \in [k]}$ and a lower bound of $\alpha k - 1$ on their sum.
Now, we move to the main step of our proof: we will use these two bounds 
to assess the total number of representatives of the voters from $N'$ in top $k$ positions.

Let $z_i(t)=1$ if $\delta_i(t)$ is positive and $z_i(t)=0$ otherwise.
Note that if $z_i(t)=1$, this means that  
voter $i$ gets one more representative at step $t$.
It follows that $\avg(N', r_{\le k}) \ge \frac{1}{n'}\sum_{i \in N'}\sum_{t=1}^k z_i(t)$. 

Recall that the Cauchy-–Schwarz inequality states that for every pair of sequences of $n$ real values, 
$(u_1, \ldots, u_\nu)$ and $(v_1, \ldots, v_\nu)$ it holds that 
\begin{align}\label{ineq:cauchy–schwarz}
\left(\sum_{i=1}^\nu u_i v_i\right)^2\leq \left(\sum_{i=1}^\nu u_i^2\right) \left(\sum_{i=1}^\nu v_i^2\right)\text{.}
\end{align}
Applying this inequality, we get 
\begin{align*}
\left(\sum_{i \in N'}\sum_{t=1}^k \delta_i(t)\right)^2 =
\left(\sum_{i \in N'}\sum_{t=1}^k z_i(t) \delta_i(t)\right)^2 \leq 
\left(\sum_{i \in N'}\sum_{t=1}^k z_i^2(t)\right) \left(\sum_{i \in N'}\sum_{t=1}^k \delta_i^2(t)\right)\text{.}
\end{align*}
We infer that $z = \sum_{i \in N'}\sum_{t=1}^k z_i^2(t) \geq \frac{(\alpha k-1)^2}{3ky}$.
Now, recall that $\frac{5\lambda}{\alpha^2} + \frac{1}{\alpha}\le k \le \frac{5\lambda}{\alpha^2} + \frac{1}{\alpha} +1$,
and hence $\alpha k-1\ge \frac{5\lambda}{\alpha}$.
We have
\begin{align*}
\frac{z}{n'} \geq \frac{(\frac{5\lambda}{\alpha})^2}{3k} \geq 
\frac{(\frac{5\lambda}{\alpha})^2}{3(\frac{5\lambda}{\alpha^2} + \frac{1}{\alpha}+1)} \geq 
\frac{\frac{25\lambda^2}{\alpha^2}}{3 \cdot \frac{7\lambda}{\alpha^2}} \geq \lambda \text{.}
\end{align*}
Here the inequality $\frac{5\lambda}{\alpha^2} + \frac{1}{\alpha} +1 \leq \frac{7\lambda}{\alpha^2}$ 
follows from the fact that $\alpha \leq 1$, $\lambda\ge 1$. 
This gives a contradiction and completes the proof.

Now, let us provide some intuition behind the mathematical formulas and explain why it was useful to consider 
the sum of squares, $S(t)$. Phragm\'{e}n's rule aims at distributing the load among the voters as equally as 
possible. Intuitively, to show that on average the voters have a significant number of representatives, one needs 
to show that, on average, when a voter gets an additional representative, she is assigned a relatively small amount 
of load. In some sense $S(i)$ can be viewed as a potential function. In each step, this `potential function' 
increases by a bounded value. Yet, since the quadratic function is convex, when a voter gets an additional 
representative, the potential function drops superlinearly with 
respect to the load $\ell$. This allows to infer that the increase of the load 
of a voter getting an additional representative 
can be bounded; formally, this is accomplished by using the Cauchy–-Schwarz inequality. Considering the sum of squares 
of the values $e_{i}(t)$ allows us to invoke this inequality, yet we believe that 
similar bounds can be obtained by considering other types of convex transformations of the values 
$e_{i}(t)$.
\end{proof}

Let us recall the relation between arithmetic, geometric, and harmonic means. For each sequence of positive values, $a_1, \ldots a_n$, it holds that:
\begin{align}\label{eq:inequ}
\frac{1}{n}\sum_{i=1}^n a_i \;\geq\; \sqrt[n]{\prod_{i=1}^n a_i} \;\geq\; \frac{n}{\sum_{i=1}^n\frac{1}{a_i}} \text{.}
\end{align}

\begin{theorem}\label{thm:seq-pav}
SeqPAV satisfies $\kappa(\alpha, \lambda)$-group representation for 
$\kappa(\alpha, \lambda) = \big\lceil\frac{2(\lambda+1)^2}{\alpha^2}\big\rceil$.
\end{theorem}
\begin{proof}
Fix $\alpha\in (0, 1]$, $\lambda\in{\mathbb N}$, a profile $P$,
and a group of voters $N' \subseteq N$ such that $n' = |N'| = \lceil \alpha n \rceil$
and $\coh(N')\ge \lambda$. 
Set $k = \big\lceil \frac{2(\lambda+1)^2}{\alpha^2} \big\rceil$.
Let $r$ be the ranking returned by SeqPAV. 

Assume for the sake of contradiction that $\avg(N', r_{\le k}) < \lambda$ 
and set $z= n'\cdot \avg(N', r_{\le k})$; note that $z < \lambda n'$.
By our assumption, 
at the end of each step $t\in[k]$ there exists some alternative that is approved by all voters in $N'$,
but has not been ranked yet; let $h$ be some such alternative.
For $t\in[k+1]$, consider the moment just before the $t$-th step of SeqPAV.
Let $a_i(t)$ denote the number of alternatives selected so far that appear in $A_i$,
and let $T(t) = \sum_{i \in N}\frac{1}{a_i(t) + 1}$. Note that 
$n=T(1)\ge T(2)\ge \ldots \ge T(k+1)\ge 0$.
For each $t\in[k+1]$ we have
\begin{align*}
\textstyle\sum_{i \in N'}a_i(t) \le z \textrm{.}
\end{align*} 
From \eqref{eq:inequ} we infer
\begin{align*}
\frac{z + n'}{n'} \ge \frac{\sum_{i \in N'} (a_i(t) + 1)}{n'} \geq 
\frac{n'}{\sum_{i \in N'} \frac{1}{a_i(t)+1}}\textrm{,}
\end{align*}
that is,
\begin{align*}
\sum_{i \in N'} \frac{1}{a_i(t)+1} \ge \frac{(n')^2}{z + n'} > \frac{(n')^2}{n'(\lambda + 1)} = 
\frac{n'}{\lambda+1}\text{.} 
\end{align*}
Consider the alternative $a$ selected by SeqPAV at step $t$, $t\in[k]$. 
Without loss of generality, assume that $N_a=\{1, \dots, s\}$
at step $t$ SeqPAV selects an alternative $a$ that is approved by voters $1, \ldots, s$.
Since alternative $h$ is available at this step,  
it has to be the case that SeqPAV favors $a$ over $h$, \ie
for the harmonic weight vector $\vecw=(1, \nicefrac12, \nicefrac13, \dots)$
we have  
$w(r_{\le t-1}\cup\{a\}) - w(r_{\le t-1}) \ge w(r_{\le t-1}\cup\{h\}) - w(r_{\le t-1})$.
This implies
\begin{align*}
\sum_{i=1}^{s}\frac{1}{a_i(t) + 1} \ge \sum_{i \in N'} \frac{1}{a_i(t)+1} > \frac{n'}{\lambda+1} \text{.}
\end{align*}
Thus, we have
\begin{align*}
T(t) - T(t+1) &= \sum_{i=1}^{s}\Big(\frac{1}{a_i(t) + 1} - \frac{1}{a_i(t) + 2}\Big) \\
              &= \sum_{i=1}^{s}\frac{1}{(a_{i}(t)+1)(a_{i}(t) + 2)} \ge \sum_{i=1}^{s}\frac{1}{2(a_{i}(t) + 1)^2}.
\end{align*}
Applying the Cauchy--Schwarz inequality to $(\frac{1}{a_1(t)+1}, \dots, \frac{1}{a_s(t)+1})$ and $(1, \dots, 1)$,
we obtain
\begin{align*}
&\sum_{i=1}^{s} \frac{1}{2(a_{i}(t) + 1)^2} \geq \frac{1}{2s}\Big(\sum_{i=1}^{s}\frac{1}{a_{i}(t) + 1}\Big)^2 
                      > \frac{1}{2n} \Big(\frac{n'}{\lambda + 1} \Big)^2 \geq \frac{n \alpha^2}{2(\lambda+1)^2} \textrm{.}
\end{align*}
Since the above inequality holds for each $t\in[k]$, we have
\begin{align*}
T(1) - T(k+1) &= \sum_{t = 1}^{k}\big( T(t) - T(t+1) \big) > \frac{kn \alpha^2}{2(\lambda+1)^2} \geq n\text{.}
\end{align*}
Thus $T(k+1) < T(1) - n = n - n = 0$, 
a contradiction. This completes the proof.
\end{proof}

The technique developed in the proof of Theorem~\ref{thm:seq-pav} 
can be used to provide similar bounds for other RAV rules.
In particular, for the $p$-geometric rule we obtain the following bound. 

\begin{theorem}\label{thm:pgeometric}
For each $p > 1$, the $p$-geometric rule satisfies 
$\kappa(\alpha, \lambda)$-group representation for 
$\kappa(\alpha, \lambda) = \big\lceil\frac{p^{\lambda+1}}{\alpha(p-1)}\big\rceil$.
\end{theorem}
\begin{proof}
We will use the same notation as in the proof of
Theorem~\ref{thm:seq-pav}.
Fix $\alpha\in (0, 1]$, $\lambda\in{\mathbb N}$, a profile $P$,
and a group of voters $N' \subseteq N$ such that $n' = |N'| = \lceil \alpha n \rceil$
and $\coh(N')\ge \lambda$.
Set $k = \big\lceil \frac{p^{\lambda+1}}{\alpha(p-1)} \big\rceil$.
Let $r$ be the ranking returned by the $p$-geometric rule.

Assume for the sake of contradiction that $\avg(N', r_{\le k}) < \lambda$
and set $z= n'\cdot \avg(N', r_{\le k})$; note that $z < \lambda n'$.
By our assumption,
at the end of each step $t\in[k]$ there exists some alternative that is approved by all voters in $N'$,
but has not been ranked yet; let $h$ be some such alternative.
For $t\in[k+1]$, consider the moment just before the $t$-th step of the $p$-geometric rule.
Let $a_i(t)$ denote the number of alternatives selected so far that appear in $A_i$,
and let $T(t) = \sum_{i \in N}\frac{1}{p^{a_{i}(t)}}$. 
Note that $n=T(1)\ge T(2)\ge \ldots \ge T(k+1)\ge 0$.
We have
\begin{align*}
\textstyle\sum_{i \in N'} a_i(t) \le z \textrm{.}
\end{align*} 
From the inequality between the geometric mean and the harmonic mean we infer
\begin{align*}
\sum_{i \in N'} \frac{1}{p^{a_i(t)+1}} 
\geq \frac{n'}{\sqrt[n']{\prod_{i \in N'}p^{a_i(t)+1}}} \geq \frac{n'}{p^{\frac{z + n'}{n'}} } 
> \frac{n'}{p^{\lambda +1}} \textrm{.}
\end{align*}
Consider the alternative $a$ selected by the $p$-geometric rule at step $t$, $t\in[k]$;
assume without loss of generality that $N_a=\{1, \dots, s\}$.
By our assumption, alternative $h$ is still available at that step.
This means that for $\vecw=(1, \nicefrac{1}{p}, \dots)$ we have
$w(r_{\le t-1}\cup\{a\}) - w(r_{\le t-1}) \ge w(r_{\le t-1}\cup\{h\}) - w(r_{\le t-1})$.
This means that
$$
\sum_{i=1}^{s} \frac{1}{p^{a_i(t)+1}} \ge \sum_{i\in N'} \frac{1}{p^{a_i(t)+1}} > \frac{n'}{p^{\lambda +1}},
$$
and thus we have
\begin{align*}
&T(t) - T(t+1) = \sum_{i=1}^{s}\Big(\frac{1}{p^{a_i(t)}} - \frac{1}{p^{a_i(t)+1}}\Big) \\
 &\quad = \sum_{i=1}^{s} \frac{p-1}{p^{a_i(t)+1}} = (p-1) \sum_{i=1}^{s} \frac{1}{p^{a_i(t)+1}} \\
 &\quad >  (p-1)\frac{n'}{p^{\lambda+1}} \geq n \frac{\alpha(p-1)}{p^{\lambda+1}} 
\geq \frac{n}{k} \text{.}
\end{align*}
Thus, we obtain
\begin{align*}
T(k+1) < T(1) - k \cdot \frac{n}{k} = n - n = 0 \textrm{.}
\end{align*}
We reached a contradiction, which completes the proof.
\end{proof}

Theorem~\ref{thm:pgeometric} establishes 
a linear relationship between the proportion of the group $\alpha$ and the guarantee 
$\kappa(\alpha, \lambda)$. Thus, our bound for the $p$-geometric rule
is better than our bounds for Phragm\'{e}n's rule and SeqPAV
from Theorems~\ref{thm:phragmen} and~\ref{thm:seq-pav}, respectively. Further, as suggested 
by Example~\ref{ex:lower_bound}, a linear relationship is the best we can hope for. 
Unfortunately, as a tradeoff we obtain 
an exponential relationship between the required amount of representation $\lambda$ 
and the guarantee $\kappa(\alpha, \lambda)$. Nevertheless, 
Theorem~\ref{thm:pgeometric} shows that
if we are only interested in optimizing 
$\kappa(\alpha, \lambda)$ for a \emph{constant} value of $\lambda > 1$, the $\big(\frac{\lambda+1}{\lambda}\big)$-geometric rule 
provides very good guarantees for group representation.

%Observe that for a constant $x > 1$ the function $f(p) = \frac{p^{x+1}}{p-1}$ is minimized for $p = \frac{x+1}{x}$.

\begin{corollary}\label{cor:geometric}
For a constant $\lambda\in{\mathbb N}$, 
the $\big(\frac{\lambda+1}{\lambda}\big)$-geometric rule satisfies 
$\kappa(\alpha, \lambda)$-group representation for $\kappa(\alpha, \lambda) = \lceil\frac{e(\lambda+1)}{\alpha}\rceil$.
\end{corollary}
\begin{proof}
For $p = \frac{\lambda+1}{\lambda}$ the formula for $\kappa(\alpha, \lambda)$ 
from Theorem~\ref{thm:pgeometric} can be rewritten as follows:
\begin{align*}
\frac{p^{\lambda+1}}{\alpha(p-1)} = 
\frac{\big(\frac{\lambda+1}{\lambda}\big)^{\lambda+1}}{\alpha\big(\frac{\lambda+1}{\lambda}-1\big)} 
\leq \frac{e \cdot \frac{\lambda+1}{\lambda}}{\alpha\cdot \frac{1}{\lambda}} = \frac{e(\lambda+1)}{\alpha} \text{.}
\end{align*}
\end{proof}

For Reverse SeqPAV, we have not been able to establish an analogue of Theorems~\ref{thm:phragmen}, \ref{thm:seq-pav}
and~\ref{thm:pgeometric}.
However, we can obtain a bound on $\kappa(\alpha, \lambda)$ for each $\alpha\in(0, 1]$ and for \emph{some} sufficiently 
large $\lambda$.

Our proof is based on two observations, formalized as 
Lemma~\ref{lemma:rev-rav1}~and~\ref{lemma:rev-rav2}, below.

\begin{lemma}\label{lemma:rev-rav1}
Consider a group of voters, $N'$, and a set of alternatives, $A'\subseteq A$, with a total approval score of the voters $N'$ from 
$A'$ equal to $z$, i.e., $\sum_{i\in N'}|A_i\cap A|=z$. Assume that Reverse SeqPAV removes one alternative from $A'$ and this 
results in decreasing the total approval score of $N'$ to $z' < z$. Then it holds that removing any alternative 
decreases the PAV-score by at least $\frac{(z - z')^2}{z}$.
\end{lemma}
\begin{proof}
Since the total satisfaction of voters from $N'$ decreases by $z - z'$, it means that there are $z - z'$ voters in $N'$ 
who, after removing the alternative, lose one of their representatives. Let us rename these voters to $i_1$, $\ldots$ 
$i_{z - z'}$. Let $a(i)$ denote the number of representatives of the $i$-th voter just before removing the alternative 
from $A'$. It holds that:
\begin{align*}
\frac{z}{z - z'} \geq \frac{\sum_{j = 1}^{z-z'} a(i_j)}{z - z'} \geq 
\frac{z - z'}{\sum_{j = 1}^{z-z'} \frac{1}{a(i_j)}}\textrm{,}
\end{align*}
where the last inequality follows from the inequality between the arithmetic and the harmonic means.
After reformulating the above inequality, we get that:
\begin{align*}
\sum_{j = 1}^{z-z'} \frac{1}{a(i_j)} \geq \frac{(z - z')^2}{z}\textrm{.}
\end{align*}
Thus, as a result of removing the alternative from $A'$, the PAV-score of the voters decreases by at least $\frac{(z - 
z')^2}{z}$. The statement of the lemma holds since Reverse SeqPAV selects an alternative 
that decreases the PAV-score of the voters the least.
\end{proof}

\begin{lemma}\label{lemma:rev-rav2}
Consider a set of alternatives, $A'\subseteq A$. Assume that Reverse SeqPAV removes an 
alternative from $A'$ and this decreases the PAV-score of the voters by $\Delta$. It holds that $|A| \leq \nicefrac{n}{\Delta}$.
\end{lemma}
\begin{proof}
Let $a(i)$ denote the number of representatives of the $i$-th voter just before removing the alternative from $A'$. 
Consider an alternative $y \in A$. Let $i_1$, $\ldots$, $i_{p}$ be the voters that approve of $y$. Since Reverse SeqPAV
selects an alternative that decreases the PAV-score of the voters the least, it 
holds that:
\begin{align*}
\sum_{j = 1}^p \frac{1}{a(i_j)} \geq \Delta.
\end{align*}
Summing up the left sides of the above inequality over all alternatives from $A'$ we get that:
\begin{align*}
\sum_{i \in N} \sum_{j = 1}^{a(i)} \frac{1}{a(i)} \geq |A|\Delta.
\end{align*}
Which gives that:
\begin{align*}
\sum_{i \in N} \sum_{j = 1}^{a(i)} \frac{1}{a(i)} = \sum_{i \in N} 1  = n \geq |A|\Delta,
\end{align*}
which completes the proof.
\end{proof}

\begin{theorem}\label{thm:rev-seq-pav}
Let $\alpha \in (0, 1]$, $\lambda \in \naturals$, 
and let $N' \subseteq N$ be an $(\alpha,\lambda)$-significant group.
% a group of voters such that $n' = |N'| = \lceil \alpha n \rceil$, $\coh(N')\ge \lambda$.
Let $r$ be the ranking returned by Reverse SeqPAV.
Then, there exists $y \geq \lambda$ such that
$\avg(N', r_{\le y/\alpha})\ge y$.
%such that on average $y$ out of the first $\frac{y}{a}$ alternatives 
%in the ranking returned by Reverse Sequential Proportional Approval Voting are approved by the members of $N'$.
\end{theorem}
\begin{proof}
Consider a set of voters 
$N' \subseteq N$ with $|N'|=n'$, $\coh(N')\ge \lambda$. 
Let $W$, $|W|=\lambda$, be a set of alternatives approved by all members of $N'$.
Let us consider the moment when Reverse SeqPAV takes the first 
alternative from $W$ and puts it in position $k$ in the ranking. 
Let $y$ denote the average satisfaction of the voters, at this point; naturally, $y \geq \lambda$.

The total approval score of the voters in $N'$ gained from $r_{\le k}$ is at least equal to $n'y$, 
and by Lemma~\ref{lemma:rev-rav1} we infer 
$w_{\textrm PAV}(r_{\le k}) - w_{\textrm PAV}(r_{\le k}) \ge 
\frac{(n')^2}{n'y} = \frac{n'}{y}$. By Lemma~\ref{lemma:rev-rav2} with $A'=r_{\le k}$, we get 
that $k \leq \frac{ny}{n'} \leq \frac{y}{\alpha}$ and since $k$ is an interger $k \le \lfloor y/\alpha\rfloor$. This completes the proof.
\end{proof}

%%%%%%%%%%%%%%%%%%%%%%%%%%%%%%%%%%%%%%%%%%%%%%%%%%%%%%%%%%%%%%%%%%%%%%%%%%%%%%%%%%%%%%%%%%%%%%

\section{Experimental Evaluation of Ranking Rules}

The results from the previous section provide worst-case guarantees for several interesting ranking rules. We would now like to complement these results with upper bounds, i.e., observed ``violations'' of the justified demand of voter groups. To this end, we consider a large number of synthetic and real-world preference data sets and analyze the representation offered by ranking rules. 
%In the following we write $[x:y]$ to denote that the number of candidates or voters ranges from $x$ to $y$ and is chosen uniformly at random for each instance. 
% Our goal was to consider data sets of very different kinds so that strengths and weaknesses of ranking rules become apparent; hence the following---very diverse---probability distributions and data sets are considered.
In order to illuminate strengths and weaknesses of the ranking rules, the following very diverse probability distributions and data sets are considered.
In total, our experiments are based on 315,500 instances.

\begin{description}
\item[Random subsets.] In this model, votes are random subsets of the set of alternatives with 
the number of alternatives ranging from $4$ to $14$ and the number of voters ranging from $3$ to $300$.
We distinguish small profiles with $4\leq m\leq 6$ and $3\leq n \leq 10$, and large profiles with $9\leq m\leq 14$ and $20\leq n \leq 300$.
%300.000 ``small'' profile with $[4:6]$ alternatives and $[3:10]$ voters, each approving $[2:|C|-1]$ alternatives,
%6.000 larger profiles with $[9:14]$ alternatives and $[20:250]$ voters, each approving $[4:8]$ alternatives, and
%3.000 larger profiles with $[9:14]$ alternatives and $[50:300]$ voters, each approving $x$ alternatives where $x$ is fixed for each profile to be a number in $[3:|C|-1]$.

\item[Spatial Model with Districts.] In this model, we represent voters and alternatives as points in two-dimensional Euclidean space $[0, 1] \times [0, 1]$. In this space, we first place three \emph{districts} by randomly selecting a center point for each district.
% Each district has 5 alternatives.
%; the three districts have respectively 200, 300, and 500 voters. 
Each district defines a Gaussian distribution over $[0, 1] \times [0, 1]$, centered at the district center point with a standard deviation of 0.2 in both dimensions. For each district, we then sample a number of points representing voters and alternatives according to this distribution. Each voter approves of all alternatives within a radius of 0.4.

%\item[Mallows Model with Districts.] We set the number of alternatives to 15. The Mallows model defines a distribution over linear orders. 
%In the standard Mallows model~\citep{mal:j:mallows} we are given two parameters: a central preference ranking $r_c$ and a fraction $\phi \in [0, 1]$. The probability of drawing a ranking $r$ is proportional to the value $\phi^{d_{K}(r_c, r)}$, where $d_K$ is the Kendal--Tau distance~\citep{ken-gib:b:rank-methods}. In our simulations we generated sets of 500 to 1200 voters split into four districts. For each district $d$ we selected a random central preference ranking $r_d$; for each voter in the district we sampled a ranking from the Mallows distribution with parameters $r_c = r_d$ and $\phi = \nicefrac{1}{2}$, and we assumed that such voter approves of the first 8 alternatives in its ranking.

\item[Urn Model.] We consider $9 \le m \le 15$ alternatives and 200 to 600 voters. For each voter, we sample a \emph{ranking} according to the Polya--Eggenberger Urn Model~\citep{bpublicchoice85}, and turn this ranking into an approval set by letting the voter approve the first 5 to 8 alternatives of the ranking. The rankings are sampled as follows. Consider an urn that initially contains each of the $m!$ possible rankings. The first voter's ranking is picked uniformly at random from the urn. To achieve some correlation between the rankings of different voters, we then insert~$b$ copies ($b>1$) of the selected ranking into the urn (we use $b = 0.05 \cdot m!$). Then the second voter's ranking is picked from the urn, again $b$ copies are added to it, and so on.

\item[Two groups.] In this model, we randomly assign voters into two groups, where all members of the same group approve the same alternatives. The approved alternatives of these two groups may overlap.

\item[Real-world data.] We consider 346 real-world preference profiles from PrefLib~\citep{mat-wal:c:preflib} consisting of rankings with ties with $m\leq 25$ and $n\leq 2000$. 
The approval sets of voters consist of a certain number of their top-ranked alternatives (the concrete number of approved alternatives in most cases varies between 1 and $\nicefrac{m}{2}$; sometimes it exceeds $\nicefrac{m}{2}$ due to ties).
\end{description}

\begin{figure}[ht!]
\begin{minipage}[h]{0.31\linewidth}
  \centering
  \includegraphics[width=\textwidth]{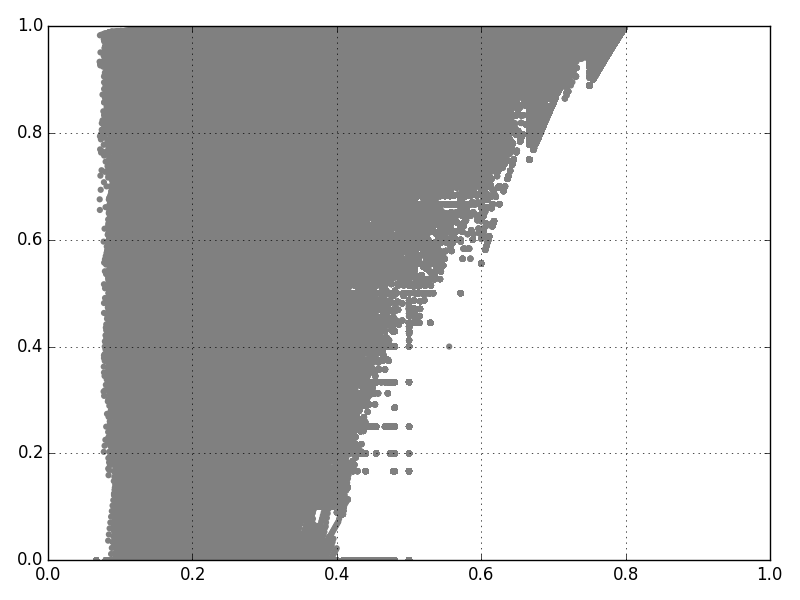}
  AV
\end{minipage}
\hspace{0.1cm}
\begin{minipage}[h]{0.31\linewidth}
  \centering
  \includegraphics[width=\textwidth]{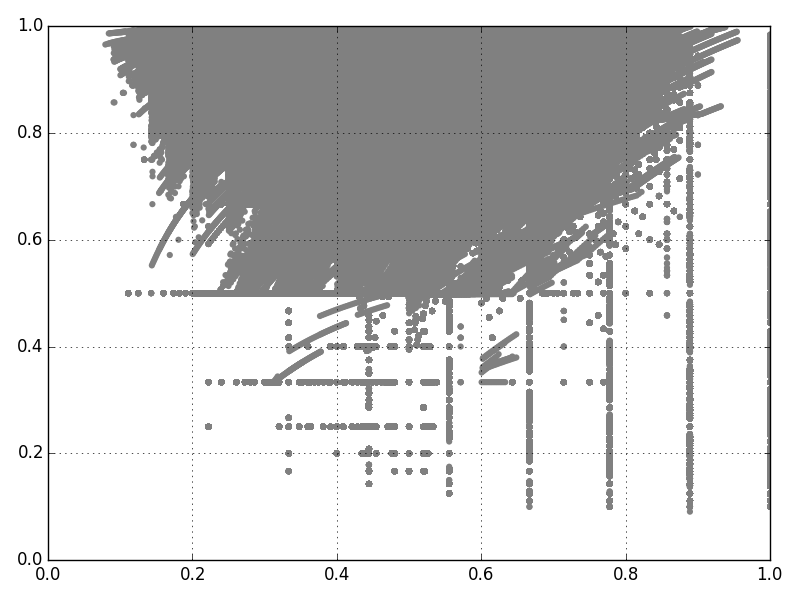}
  Greedy Chamberlin-Courant
\end{minipage}
\hspace{0.1cm}
\begin{minipage}[h]{0.31\linewidth}
  \centering
  \includegraphics[width=\textwidth]{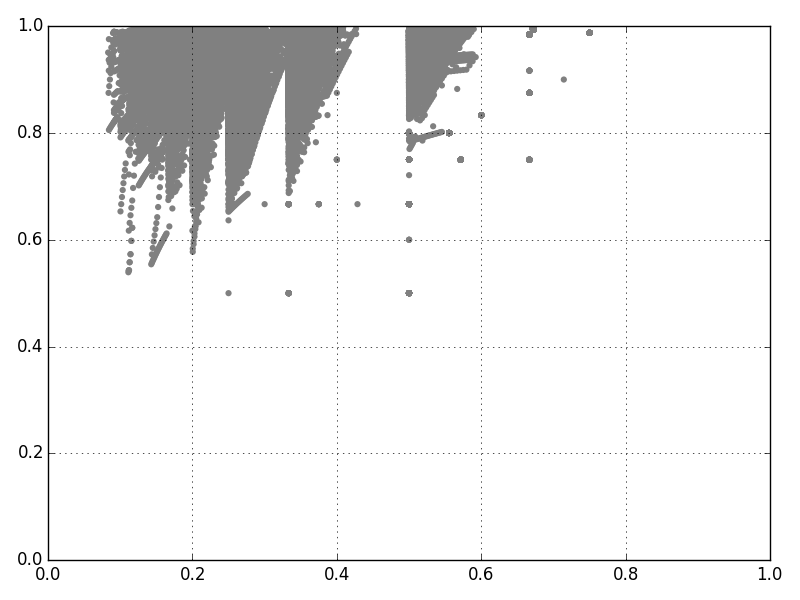}
  Phragm\'{e}n
\end{minipage}
  \vspace{0.2cm}

\begin{minipage}[h]{0.31\linewidth}
  \centering
  \includegraphics[width=\textwidth]{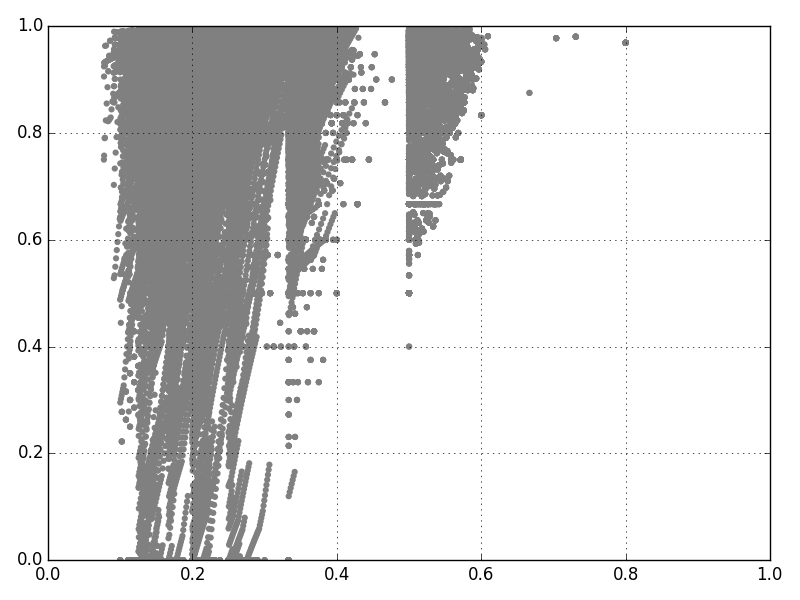}
  $\nicefrac{5}{4}$-Geometric RAV
\end{minipage}
\hspace{0.1cm}
\begin{minipage}[h]{0.31\linewidth}
  \centering
  \includegraphics[width=\textwidth]{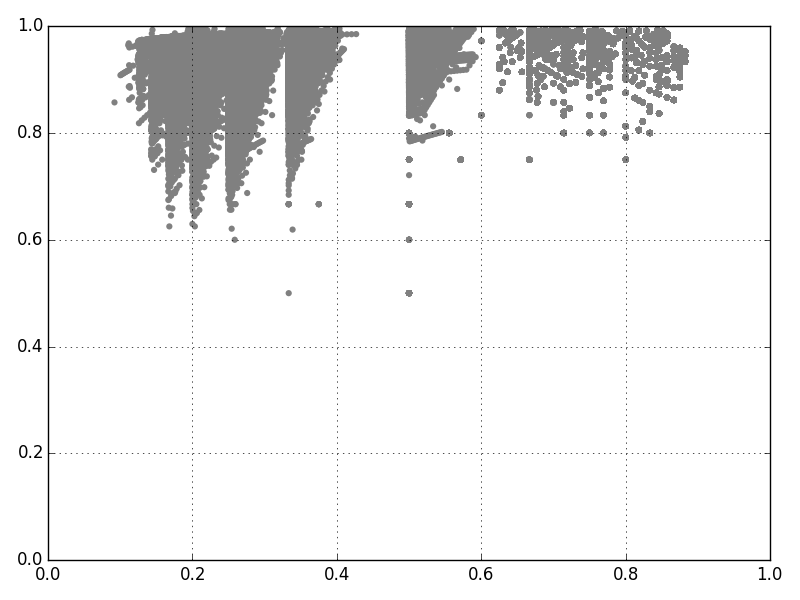}
  $2$-Geometric RAV
\end{minipage}
\hspace{0.1cm}
\begin{minipage}[h]{0.31\linewidth}
  \centering
  \includegraphics[width=\textwidth]{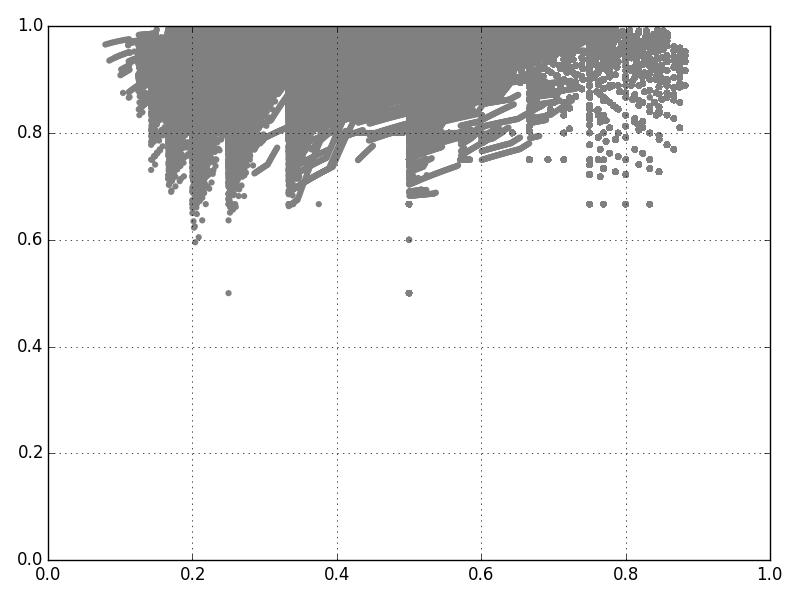}
  $10$-Geometric RAV
\end{minipage}

  \vspace{0.2cm}

\begin{minipage}[h]{0.31\linewidth}
  \centering
  \includegraphics[width=\textwidth]{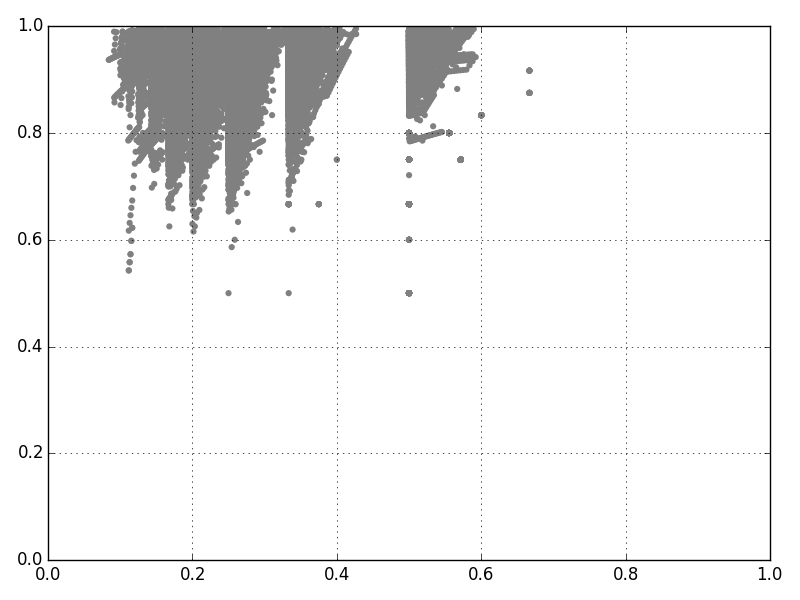}
  SeqPAV
\end{minipage}
\hspace{0.1cm}
\begin{minipage}[h]{0.31\linewidth}
  \centering
  \includegraphics[width=\textwidth]{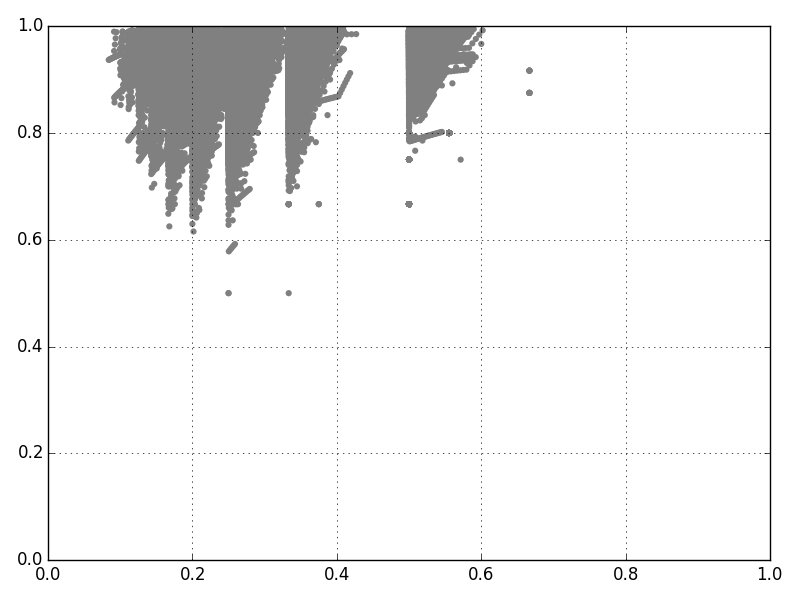}
  Reverse SeqPAV
\end{minipage}
\hspace{0.1cm}
\begin{minipage}[h]{0.31\linewidth}
  \centering
  \includegraphics[width=\textwidth]{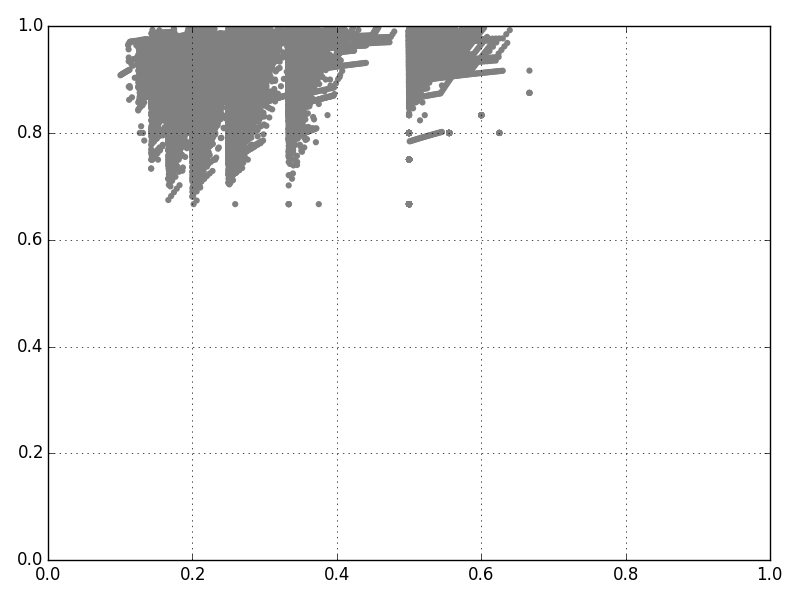}
  Best-of
\end{minipage}

\caption{Violations of justified demand encountered in our data sets. The x-axis shows the proportion of the respective group of voters, $\alpha(N')$, the y-axis the quotient of average representation and justified demand, $y_r(N',k)$. Roughly, ranking rules with more gray points perform worse according to our metrics; lower points correspond to more severe violations of proportionality. Points to the left correspond to small groups, points to the right correspond to large groups with unmet justified demand.}
\label{fig:scatter}
\end{figure}

\subsubsection{Measures of Quality of Group Representation} 

In our experiments we record every \emph{violation} of justified demand:
If, for a ranking $r$ and an integer $k\leq m$, there exists a group $N'$ with 
$\avg(N',r_{\le k}) < \jar(N',k)$, we plot a point at $(\alpha(N'), y_r(N',k))$ indicating this violation, 
where $y_r(N',k)=\frac{\avg(N',r_{\le k})}{\jar(N',k)}$.
Figure~\ref{fig:scatter} shows these plots for different ranking rules.
Violations displayed in the lower part of the plots have a small ratio $y_r$ and thus are more severe.
Note that several points may originate from the same ranking, and that different rankings may produce the same point. Hence, these plots do not display \emph{how often} violations occur but rather \emph{in which regions} (small/large groups, minor/major violations) violations have been recorded.

\subsubsection{Results of the Simulations}
Let us start by analysing the plots of Figure~\ref{fig:scatter}.
Approval Voting (AV) and greedy Chamberlin--Courant (greedy CC) do not do well:
while AV at least provides reasonable representation for large groups (consistent with \thmref{thm:av}), greedy CC produces violations all across the spectrum.
This is not too surprising since greedy CC only cares about representing each voter by a single alternative, and selects alternatives arbitrarily once this is achieved.
We consider three geometric RAV rules, for values of $p$ in $\{\nicefrac{5}{4}, 2, 10\}$.  
The $\nicefrac{5}{4}$-geometric RAV rule has characteristics similar to AV (the second and third approved alternatives count almost as much as the first), while the $10$-geometric RAV rule is similar to greedy CC (the first approved alternative counts the most).
The $2$-Geometric RAV performs best, together with SeqPAV, Reverse SeqPAV, and Phragm\'{e}n's rule; it is hard to visually compare these rules with each other.
We also added a ``best-of'' rule, which selects whatever ranking $r$ has the highest quality $q_P(r)$ out of those rankings generated by our rules (but note that this is not the optimal ranking according to \defref{def:opt}).

The best rules according to~\ref{fig:scatter} are thus 2-Geometric RAV, SeqPAV, Reverse SeqPAV, and Phragm\'{e}n's rule. In which contexts should we prefer to use each rule? To answer this question, it is useful to compare their performance on instances obtained by the different distributions and data sets we employed.
First of all, all rankings produced by these rules have a similar quality in terms of worst-case violations: the worst violations of all these rules were for $y_r(N',k)=0.5$.
Consequently, for all instances considered these rules provided $q_P(r)\geq 0.5$.
A more discriminating measure is the percentage of instances (of a given distribution) that have been solved perfectly, i.e., without violations.
The results are summarized in Table~\ref{tab:statistics}. 
Notably, 2-geometric RAV achieves perfection most frequently among the studied rules in four categories: real-world instances, large profiles with random subsets, the urn model, and the spatial model.
Its performance in the spatial model is exceptional: it solves 82.7\% of the instances without violations; the runner-up is SeqPAV with 60.0\%.
The weakest category for 2-geometric RAV is ``two groups'', where it solves 87.9\% without violations; SeqPAV and Reverse SeqPAV solve all instances obtained from this distribution without violations.
For small profiles in the random subset category, Phragm\'{e}n's rule and Reverse SeqPAV perform best (99.3\%), whereas SeqPAV and 2-Geometric RAV perform slightly worse (99.2\% and 99.1\%, respectively).

% MB: it would be nice if this table displayed the percentage of instances that *have* been solved optimally, in order to align it with the text.

\begin{table*}
\centering
\label{my-label}
\begin{tabular}{lrrrrrrr}
\toprule
                                & real-world      & small           & large            & urn             & spatial          & 2 groups        & $\max \alpha$    \\%  & $\min$ y      & $\mathrm{mean}$ y \\
\midrule
SeqPAV                          & 5.0 \%          & 0.8 \%          & 34.3 \%          & 10.0 \%         & 40.0 \%          & \textbf{0.0 \%} & \textbf{0.67} \\%& \textbf{0.50} & \textbf{0.92}     \\
AV                              & 9.8 \%          & 4.6 \%          & 40.0 \%          & 16.2 \%         & 79.5 \%          & 82.9 \%         & 0.80          \\%& 0.00          & 0.54              \\
$\nicefrac{5}{4}$-geom.\ RAV & 6.6 \%          & 1.5 \%          & 36.1 \%          & 13.0 \%         & 62.8 \%          & 22.4 \%         & 0.80          \\%& 0.00          & 0.85              \\
$2$-geom.\ RAV               & \textbf{4.9 \%} & 0.9 \%          & \textbf{33.8 \%} & \textbf{8.1 \%} & \textbf{17.3 \%} & 12.1 \%         & 0.88          \\%& \textbf{0.50} & \textbf{0.92}     \\
$10$-geom.\ RAV              & 11.6 \%         & 0.9 \%          & 42.3 \%          & 29.9 \%         & 72.5 \%          & 13.9 \%         & 0.88          \\%& \textbf{0.50} & \textbf{0.92}     \\
Phragm\'en                   & 6.0 \%          & \textbf{0.7 \%} & 35.5 \%          & 9.7 \%          & 42.7 \%          & 0.2 \%          & 0.75          \\%& \textbf{0.50} & \textbf{0.92}     \\
Rev.\ SeqPAV                  & 5.0 \%          & \textbf{0.7 \%} & 36.9 \%          & 10.7 \%         & 41.2 \%          & \textbf{0.0 \%} & \textbf{0.67} \\%& \textbf{0.50} & \textbf{0.92}     \\
Greedy CC                       & 42.6 \%         & 20.4 \%         & 47.1 \%          & 68.3 \%         & 91.3 \%          & 69.8 \%         & 1.00          \\%& 0.09          & 0.69              \\
best-of                         & 4.0 \%          & 0.3 \%          & 30.2 \%          & 6.3 \%          & 13.3 \%          & 0.0 \%          & 0.67     \\     %& 0.67          & 0.92             
\bottomrule
\end{tabular}
\caption{Percentage of profiles with $q_P(r)<1$ for different categories of datasets: Random Subsets with small profiles (\emph{small}), Random Subsets with large profiles (\emph{large}), Urn Model (\emph{urn}), Spatial Model (\emph{spatial}), and Two Groups (\emph{2-groups}). Column $\max \alpha$ shows the proportion $\alpha(N')$ of the largest group $N'$ with unfulfilled justifiable demand.}
\label{tab:statistics}
\end{table*}

The main strength of (Reverse) SeqPAV is the quality of their ranking for large groups: neither of them has any violations for groups $N'$ with $\alpha(N')> \nicefrac{2}{3}$.
For Phragm\'{e}n's rule this value is $0.8$ and for 2-Geometric RAV it is $0.88$ (it is also visible in Figure~\ref{fig:scatter} that 2-Geometric RAV has more violations for large groups than, e.g., SeqPAV).
%We refer the reviewer to Table~\ref{tab:statistics} for more details and numerical values.

In conclusion, our experiments indicate that (i) $2$-Geometric RAV, SeqPAV, Reverse SeqPAV, and Phragm\'{e}n's rule are the best-suited rules to generate proportional rankings among those considered, and 
(ii) there is no single best among these four rules (the best-of rule outperforms all of them).
Unfortunately, the best-of rule is certainly not practical, as it is very expensive to compute~$q_P(r)$.
Further experiments and theoretical results are required to determine which (polynomial-time computable) rule is the best choice (for a given data set).

\section{Conclusions}

In this paper, we have formalized a fundamental problem that appears in many real-life applications: \emph{proportional rankings} can provide diversified search results, can accommodate different types of users in recommendation systems, can support decision-making processes under liquid democracy, and can even produce committees with an internal hierarchical structure. Our formalization of this problem allows us to leverage classical techniques from social choice and political science to these modern application scenarios, and shine a new light on voting rules introduced as far back as the 19th century. 

After evaluating the proportionality of several appealing ranking rules both theoretically and experimentally, we identified four such rules that appear to perform very well in this area: $2$-Geometric Reweighted Approval Voting, Sequential Proportional Approval Voting and its reverse variant, and Phragm\'{e}n's rule. However, none of these rules is single-best, and there remains a need for an in-depth analysis to determine which rule is most applicable in which situation.

While all four of these rules are polynomial-time computable,
we have shown that the optimal rule (i.e. the rule that outputs rankings maximizing the quality measure $q_P$) is NP-hard to compute. It would be desirable to develop ways in which this rule can be computed in reasonable time for practical instances, and to search for other ranking rules that might provide an even better approximation to the optimal rule than the rules we have identified in this work.

\bibliographystyle{aaai}
% \bibliography{../bib/abb,../bib/group,../bib/main}
\bibliography{../../bib/abb,../../bib/group,../../bib/main}

\begin{thebibliography}{}

\bibitem[\protect\citeauthoryear{Aziz \bgroup et al\mbox.\egroup
  }{2015}]{ABC+15a}
Aziz, H.; Brill, M.; Conitzer, V.; Elkind, E.; Freeman, R.; and Walsh, T.
\newblock 2015.
\newblock Justified representation in approval-based committee voting.
\newblock In {\em Proceedings of the 29th AAAI Conference on Artificial
  Intelligence (AAAI)},  784--790.
\newblock AAAI Press.

\bibitem[\protect\citeauthoryear{Balinski and Young}{1982}]{BaYo82a}
Balinski, M., and Young, H.~P.
\newblock 1982.
\newblock {\em Fair Representation: {M}eeting the Ideal of One Man, One Vote}.
\newblock Yale University Press.
\newblock (2nd Edition [with identical pagination], Brookings Institution
  Press, 2001).

\bibitem[\protect\citeauthoryear{Behrens \bgroup et al\mbox.\egroup
  }{2014}]{BKNS14a}
Behrens, J.; Kistner, A.; Nitsche, A.; and Swierczek, B.
\newblock 2014.
\newblock {\em The Principles of LiquidFeedback}.

\bibitem[\protect\citeauthoryear{Berg}{1985}]{bpublicchoice85}
Berg, S.
\newblock 1985.
\newblock {Paradox of voting under an urn model: The effect of homogeneity}.
\newblock {\em Public Choice} 47:377--387.

\bibitem[\protect\citeauthoryear{Gallagher}{1991}]{Gall91a}
Gallagher, M.
\newblock 1991.
\newblock Proportionality, disproportionality and electoral systems.
\newblock {\em Electoral Studies} 10(1):33--51.

\bibitem[\protect\citeauthoryear{Garey and Johnson}{1979}]{gar-joh:b:int}
Garey, M., and Johnson, D.
\newblock 1979.
\newblock {\em Computers and Intractability: {A} Guide to the Theory of
  {NP}-Completeness}.
\newblock {W. H. Freeman and Company}.

\bibitem[\protect\citeauthoryear{Hu \bgroup et al\mbox.\egroup
  }{2011}]{Hu:2011:CSI:1963405.1963412}
Hu, B.; Zhang, Y.; Chen, W.; Wang, G.; and Yang, Q.
\newblock 2011.
\newblock Characterizing search intent diversity into click models.
\newblock In {\em Proceedings of the 20th International Conference on World
  Wide Web}, WWW '11,  17--26.
\newblock New York, NY, USA: ACM.

\bibitem[\protect\citeauthoryear{Janson}{2012}]{Jans12a}
Janson, S.
\newblock 2012.
\newblock Proportionella valmetoder.
\newblock Available at \url{http://www2.math.uu.se/~svante/papers/sjV6.pdf}.

\bibitem[\protect\citeauthoryear{Kingrani, Levene, and
  Zhang}{2015}]{conf/websci/KingraniLZ15}
Kingrani, S.~K.; Levene, M.; and Zhang, D.
\newblock 2015.
\newblock Diversity analysis of web search results.
\newblock In {\em Proceedings of the {ACM} Web Science Conference, WebSci},
  43:1--43:2.

\bibitem[\protect\citeauthoryear{Laslier}{2012}]{Lasl12a}
Laslier, J.-F.
\newblock 2012.
\newblock Why not proportional?
\newblock {\em Mathematical Social Sciences} 63(2):90--93.

\bibitem[\protect\citeauthoryear{Mattei and Walsh}{2013}]{mat-wal:c:preflib}
Mattei, N., and Walsh, T.
\newblock 2013.
\newblock Preflib: A library for preferences.
\newblock In {\em Proceedings of the 3nd International Conference on
  Algorithmic Decision Theory},  259--270.

\bibitem[\protect\citeauthoryear{Monroe}{1995}]{Monr95a}
Monroe, B.~L.
\newblock 1995.
\newblock Fully proportional representation.
\newblock {\em The American Political Science Review} 89(4):925--940.

\bibitem[\protect\citeauthoryear{Mora and Oliver}{2015}]{MoOl15a}
Mora, X., and Oliver, M.
\newblock 2015.
\newblock Eleccions mitjan{\c c}ant el vot d'aprovaci{\'o}. {E}l m{\`e}tode de
  {P}hragm{\'e}n i algunes variants.
\newblock {\em Butllet{\'\i} de la Societat Catalana de Matem{\`a}tiques}
  30(1):57--101.

\bibitem[\protect\citeauthoryear{Phragm{\'e}n}{1895}]{Phra95a}
Phragm{\'e}n, E.
\newblock 1895.
\newblock {\em Proportionella val. En valteknisk studie}.
\newblock Svenska sp{\"o}rsm{\aa}l 25. Lars H{\"o}kersbergs f{\"o}rlag,
  Stockholm.

\bibitem[\protect\citeauthoryear{Pukelsheim}{2014}]{Puke14a}
Pukelsheim, F.
\newblock 2014.
\newblock {\em Proportional Representation: Apportionment Methods and Their
  Applications}.
\newblock Springer.

\bibitem[\protect\citeauthoryear{S{\'a}nchez-Fern{\'a}ndez \bgroup et
  al\mbox.\egroup }{2016}]{SFFB16a}
S{\'a}nchez-Fern{\'a}ndez, L.; Elkind, E.; Lackner, M.; Fern{\'a}ndez, N.;
  Fisteus, J.~A.; {Basanta Val}, P.; and Skowron, P.
\newblock 2016.
\newblock Proportional justified representation.
\newblock In {\em Proceedings of the 30th AAAI Conference on Artificial
  Intelligence (AAAI-16)}.

\bibitem[\protect\citeauthoryear{Santos, MacDonald, and
  Ounis}{2015}]{journals/ftir/SantosMO15}
Santos, R.~L.~T.; MacDonald, C.; and Ounis, I.
\newblock 2015.
\newblock Search result diversification.
\newblock {\em Foundations and Trends in Information Retrieval} 9(1):1--90.

\bibitem[\protect\citeauthoryear{Schulze}{2011}]{Schu11b}
Schulze, M.
\newblock 2011.
\newblock Free riding and vote management under proportional representation by
  the single transferable vote.
\newblock Available at \url{http://m-schulze.9mail.de/schulze2.pdf}.

\bibitem[\protect\citeauthoryear{Thiele}{1895}]{Thie95a}
Thiele, T.~N.
\newblock 1895.
\newblock Om flerfoldsvalg.
\newblock {\em Oversigt over det Kongelige Danske Videnskabernes Selskabs
  Forhandlinger}  415--441.

\bibitem[\protect\citeauthoryear{Wang, Luo, and Yu}{2016}]{conf/waim/WangLY16}
Wang, Y.; Luo, Z.; and Yu, Y.
\newblock 2016.
\newblock Learning for search results diversification in {T}witter.
\newblock In {\em Web-Age Information Management - 17th International
  Conference, {WAIM}},  251--264.

\bibitem[\protect\citeauthoryear{Welch, Cho, and
  Olston}{2011}]{conf/www/WelchCO11}
Welch, M.~J.; Cho, J.; and Olston, C.
\newblock 2011.
\newblock Search result diversity for informational queries.
\newblock In {\em Proceedings of the 20th International Conference on World
  Wide Web, {WWW}},  237--246.

\bibitem[\protect\citeauthoryear{Xavier}{2012}]{Xavier2012471}
Xavier, E.~C.
\newblock 2012.
\newblock A note on a maximum k-subset intersection problem.
\newblock {\em Information Processing Letters} 112(12):471--472.

\end{thebibliography}

\end{document}